\documentclass[runningheads]{llncs}

\usepackage{graphicx}
\usepackage{tikz}
\usepackage{amssymb}
\usepackage{amsmath}
\usepackage{enumitem}
\newcommand{\lr}[1]{\langle #1 \rangle}
\allowdisplaybreaks

\usepackage{lipsum}

\begin{document}
\title{On the Right Path: A Modal Logic for Supervised Learning}

\author{Alexandru Baltag\inst{1}\and Dazhu Li\inst{1,2}\and Mina Young Pedersen\inst{1}}

\authorrunning{A. Baltag  et al.}

\institute{ILLC, University of Amsterdam, The Netherlands\\ \email{\{thealexandrubaltag,minaypedersen\}@gmail.com}\and Department of Philosophy, Tsinghua University\\ \email{lidazhu91@163.com}}

\maketitle              
\begin{abstract}
Formal learning theory formalizes the process of inferring a general result from examples, as in the case of inferring grammars from sentences when learning a language. Although empirical evidence suggests that children can learn a language without responding to the correction of linguistic mistakes, the importance of Teacher in many other paradigms is significant. Instead of focusing only on learner(s), this work develops a general framework---the \textit{supervised learning game (SLG)}---to investigate the interaction between \textit{Teacher} and \textit{Learner}. In particular, our proposal highlights several interesting features of the agents: on the one hand, Learner may make mistakes in the learning process, and she may also ignore the potential relation between different hypotheses; on the other hand, Teacher is able to correct Learner's mistakes, eliminate potential mistakes and point out the facts ignored by Learner. To reason about strategies in this game, we develop a \textit{modal logic of supervised learning (SLL)}. Broadly, this work takes a small step towards studying the interaction between graph games, logics and formal learning theory.
\end{abstract}


\section{Introduction}\label{sec:introduction}
Formal learning theory formalizes the process of inferring a general result from examples, as in the case of inferring grammars from sentences when learning a language. A good way of understanding this general process is by treating it as a game played by \textit{Learner} and \textit{Teacher}. It starts with a class of possible worlds, where one of them is the actual one chosen by Teacher. Learner's aim is to get to know which one it is. Teacher inductively provides information about the world, and whenever receiving a piece of information Learner picks a conjecture from the class, indicating which one she thinks is the case. Different success conditions for Learner can be defined. In this article we require that at some finite stage of the procedure Learner decides on a correct hypothesis. This kind of learnability is known as \textit{finite identification} \cite{finiteidentification}.

Although empirical evidence suggests that children can learn a language without responding to the correction of linguistic mistakes \cite{languagelearning}, the importance of teachers in many other paradigms is significant. For instance, in the paradigm of \textit{learning from queries and counterexamples} \cite{setlearning}, Teacher has a strong influence on whether the process is successful. Moreover, results in \cite{learning} suggest that a helpful Teacher may make learning easier. In this work, instead of focusing only on Learner, we highlight the interactive nature of learning.

As noted in \cite{learning}, a concise model for characterizing the interaction between Learner and Teacher is the \textit{sabotage game (SG)}. A SG is played on a graph with a starting node and a goal node, and it goes in rounds: Teacher first cuts an edge in the graph, then Learner makes a step along one of the edges still available. Both of them win iff Learner arrives at the goal node \cite{lig}. From the perspective of formal learning theory, this step-by-step game depicts a guided learning situation. Say, a natural interpretation is the situation of theorem proving. In this case, the starting node is given by axioms, the goal node stands for the theorem to be proved, other nodes represent lemmas conjectured by Learner, and edges capture Learner's possible inferences between them. Inferring is represented by moving along those edges. The information provided by Teacher can be treated as his feedback, i.e., removing edges to eliminate wrong inferences. The success condition is given by the winning condition: the learning process has been successful if Learner reaches the goal node, i.e., proving the theorem. For the general correspondence between SG and learning models, we refer to \cite{learning}.

However, we would argue that this application of SG gives a highly restricted model of learning. For instance,

\begin{itemize}
\item[$\bullet$] Intuitively, all links in the graph are inferences conjectured by Learner, which may include mistakes. From the perspective of Learner, the wrong inferences cannot be distinguished from the correct ones. Although it is reasonable to assume that Teacher is able to do so, SG does not highlight that Learner lacks perfect information. Besides, Teacher in SG has to remove a link in each round, which is overly restrictive.
\item[$\bullet$] Links removed represent wrong inferences between lemmas. So, whether or not a link deleted occurs in Learner's current proof (i.e., the current process) is important. If the proof includes a mistake, any inference after the mistake should not make sense. However, if a potential transition having not occurred in the proof is wrong, Learner can continue with her current proof. Clearly, SG cannot distinguish between these two cases.
\item[$\bullet$] The game does not distinguish between all the various ways Learner can reach the goal. That is, as long as Learner has come to the right conclusion, the game cannot tell us whether Learner has come to this conclusion in a coherent way. Reaching the correct hypothesis by wrong transitions is not reliable. The well-known Gettier cases \cite{gettier} where one has justified true belief, but not knowledge are also examples of situations in which one wrongly reaches the right conclusion. Thus, the theory developed in \cite{learning} is subject to the Gettier problems.
\item[$\bullet$] Teacher can only \textit{delete} links to decide what Learner will not learn, and thus he only teaches what Learner has already conjectured. However, during the process of learning, `possibilities may also be ignored due to the more questionable practice if assuming that one of the theories under consideration must be true. And complexity can come to be ignored through convention or habit' (\cite{learningandphilosophy}, pp. 260). Hence, it is natural to assume that Learner may ignore the correct relation between different hypotheses.
\end{itemize}

In this paper, we therefore propose a new game, called the \textit{supervised learning game (SLG)}. This game differs from the SG on several accounts, motivated by the mentioned restrictions. Before introducing its definition, we first define some auxiliary notions. 

Let $S=\lr{w_0,w_1,...,w_n}$ be a non-empty, finite sequence. We use $e(S)$ to denote its last element. Define $Set(S):=\{\lr{w_0,w_1},\lr{w_1,w_2},...,\lr{w_{n-1},w_n}\}$. For the particular case when $S$ is a singleton, $Set(S):=\emptyset$. Besides, for any $\lr{w_i,w_{i+1}}\in Set(S)$, define $S|_{\lr{w_i,w_{i+1}}}:=\lr{w_0,w_1,...,w_u}$, where $\lr{w_u,w_{u+1}}=\lr{w_i,w_{i+1}}$ and $\lr{w_u,w_{u+1}}\not=\lr{w_j,w_{j+1}}$ for any $j<i$. Intuitively, $S|_{\lr{w_i,w_{i+1}}}$ is obtained by deleting all elements occurring after $w_u$ from $S$, where $\lr{w_u,w_{u+1}}$ is the first occurrence of $\lr{w_i,w_{i+1}}$ in $S$. Say, when $S=\lr{a,b,c,a,b}$, we have $S|_{\lr{a,b}}=\lr{a}$. Now let us introduce SLG.

\begin{definition}[SLG]\label{def-clg}
A SLG $\lr{W,R_1,R_2,\lr{s},g}$ is given by a graph $\lr{W,R_1,R_2}$, the starting node $s$ and the goal node $g$. A position of the game is a tuple $\lr{R_1^i,S^i}$. The initial position $\lr{R_1^0,S^0}$ is given by $\lr{R_1,\lr{s}}$. Round $n+1$ from position $\lr{R_1^n,S^n}$ is as follows: first, Learner moves from $e(S^n)$ to any of its $R_1$-successors $s'$; then Teacher does nothing or acts out one of the following three choices:

\begin{itemize}[leftmargin=1cm]
\item[(1).] Extend $R_1^n$ with some $\lr{v,v'}\in R_2$;
\item[(2).] Transfer $\lr{S^n,s'}$ to $\lr{S^n,s'}|_{\lr{v,v'}}$ by cutting $\lr{v,v'}$ from $Set(\lr{S^n,s'})\setminus R_2$;
\item[(3).] Delete some $\lr{v,v'}\in (R_1\setminus R_2)\setminus Set(\lr{S^n,s'})$ from $R_1$.
\end{itemize}
The new position, denoted $\lr{R_1^{n+1},S^{n+1}}$, is $\lr{R_1^n,S^n}$ (when Teacher does nothing), $\lr{R_1^{n}\cup\{\lr{v,v'}\},\lr{S^{n},s'}}$ (when he chooses (1)), $\lr{R_1^{n}\setminus\{\lr{v,v'}\},\lr{S^n,s'}|_{\lr{v,v'}}}$ (if he acts as (2)), or $\lr{R_1^{n}\setminus\{\lr{v,v'}\},\lr{S^n,s'}}$ (if he chooses (3)). It ends if Learner arrives at $g$ through an $R_2$-path $\lr{s,...,g}$ or cannot make a move, with them winning in the former case and losing in the latter.
\end{definition}

Intuitively, the clause for Learner illustrates that she cannot distinguish the links starting from the current position. The sequence $S^i$ is her current learning process, which may include mistakes; $R_{1}$ represents Learner's possible inferences between conjectures; and $R_{2}$ is the correct inferences. For any position $\lr{R_1^n,S^n}$ we have $Set(S^n)\subseteq R_1^n$. Besides, both (2) and (3) above are concerned with the case that Teacher eliminates wrong transitions, but there is an important difference. The former one concerns the case that Teacher gives Learner a counterexample to show that she has gone wrong somewhere in her current process, so Learner should move back to the conjecture right before the wrong transition. In contrast, (3) illustrates that Teacher eliminates a wrong transition conjectured that has not occurred in Learner's process yet, therefore this action does not modify Learner's current process.

From the winning condition, we know that both the players cooperate with each other. It is important to recognize that Learner's action does not conflict with her cooperative nature: she makes an effort to achieve the goal in each round. Moreover, it is not hard to see that players cannot win when there exists no $R_2$-path from the starting node to the goal node. This is reasonable, since their interaction makes sense only when the goal is learnable. The correlation between the situation of theorem proving and SLG is shown in Table \ref{tab:2}.

\begin{table}[]
\centering
\caption{Correspondence between theorem proving and supervised learning games.}
\label{tab:2}
\begin{tabular}{|p{5.7cm}|p{5.7cm}|}
\hline
\textbf{Theorem Proving} & \textbf{Supervised Learning Games} \\
\hline
Axioms & Starting node\\
Theorem & Goal node \\
Lemmas conjectured by Learner & Other states except the starting state and the goal state \\
Learner's possible inference from $a$ to $b$& $R_{1}$-edge from $a$ to $b$ \\
Correct inference from $a$ to $b$ & $R_{2}$-edge from $a$ to $b$ \\
Inferring $b$ from $a$& Transition from $a$ to $b$ \\
Proof for $a$ & $R_1$-sequence from the starting node to $a$\\
Correct proof for $a$ & $R_1$-sequence $S$ from the starting node to $a$ and $Set(S)\subseteq R_2$\\
Giving a counterexample to the inference from $a$ to $b$ in the proof $S$ & Modifying $S$ to $S|_{\lr{a,b}}$ ($\lr{a,b}\in Set(S)$)\\
Giving a counterexample to the conjectured inference from $a$ to $b$ not in the proof $S$ & Deleting $\lr{a,b}$ from $R_1$ ($\lr{a,b}\not\in Set(S)$)\\
Pointing out a potential inference from $a$ to $b$ not conjectured by Learner before & Extending $R_1$ with $\lr{a,b}$\\
\hline
\end{tabular}
\end{table}

\begin{remark}
The interpretation of SLG presented in Table \ref{tab:2} can be easily adapted to characterize other paradigms in formal learning theory, such as language learning and scientific inquiry. More generally, any single-agent games, such as solitaire and computer games, can be converted into SLG. Say, the player (Learner) does not know the correct moves well, but she knows the starting position and the goal position, and has some conjectures about the moves of the game. Besides, she can be taught by Teacher: she just attempts to play it, while Teacher instructs her positively (by revealing more correct moves) or negatively (by pointing out incorrect moves, in which case Learner may have to be moved back to the moment previous to the first incorrect move, if she made any).
\end{remark}

\begin{example}\label{example-clg}
Let us consider a simple example of SLG, as depicted in Figure \ref{figure:clg}. The starting node is $a$ and the goal node is $G$. We show that players have a winning strategy by depicting the game to play out as follows. Learner begins with moving along the only available edge to node $b$. Teacher in his turn can make $\lr{e, f}$ `visible' to Learner by adding it to $R_{1}$. Then, Learner proceeds to move along $\lr{b, c}$, and Teacher extends $\lr{b, e}$ to $R_1$. Afterwards, Learner continues on the only option $\lr{c, G}$. Although she now has already arrived at the goal node, her path $\lr{a,b,c,G}$ is not an $R_2$-sequence. So, Teacher can remove $\lr{b,c}$ moving Learner back to node $b$. Next, Learner has to move to $e$, and Teacher can delete $\lr{c,G}$ from $R_1$. Finally, Learner can arrives at $G$ in 2 steps with Teacher doing nothing. Now we have $Set(\lr{a,b,e,f,G})\subseteq R_2$, so they win. 
\end{example}

\begin{figure}
\centering
\begin{tikzpicture}
\node(a)[circle,draw,inner sep=0pt,minimum size=5mm] at (0,0.75) {$a$};
\node(a)[circle,draw,inner sep=1pt,minimum size=7mm] at (0,0.75) {};
\node(b)[circle,draw,inner sep=0pt,minimum size=5mm] at (1.5,0){$b$};
\node(c)[circle,draw,inner sep=0pt,minimum size=5mm] at (1.5,1.5){$c$};
\node(e)[circle,draw,inner sep=0pt,minimum size=5mm] at (3,0){$e$};
\node(g)[circle,draw,inner sep=0pt,minimum size=5mm] at (3,1.5){$G$};
\node(f)[circle,draw,inner sep=0pt,minimum size=5mm] at (4.5,0.75){$f$};
\draw[->](a) to node [below] {$1,2$} (b);
\draw[->](a) to node [above] {$2$} (c);
\draw[->](f) to node [above] {$1,2$}(g);
\draw[->](c) to node [above] {$1$} (g);
\draw[->](b) to node [right] {$1$} (c);
\draw[->](e) to node [below] {$2$} (f);
\draw[->](b) to node [below] {$2$} (e);
\end{tikzpicture}
\caption{A SLG game ($R_{1}$ is labelled with `1' and $R_{2}$ with `2').}
\label{figure:clg}
\end{figure}
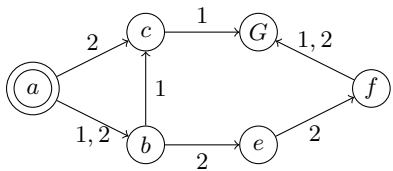

To reason about players' strategies in SLG, in what follows we will study SLG from a modal perspective. \textit{Sabotage modal logic (SML)} \cite{sabotage} is known to be a suitable tool to characterize SG, which extends basic modal logic with a sabotage modality $\lr{-}\varphi$ stating that there is an edge such that,  $\varphi$ is true at the evaluation node after deleting the edge from the model. However, given the differences between SG and SLG, we will develop a novel \textit{modal logic of supervised learning (SLL)} to capture SLG.

\medskip

\noindent \textit{Outline.} Section \ref{sec:lml} introduces SLL along with its application to SLG and some preliminary observations. Section \ref{sec:expressivepower} studies the expressivity of SLL. Section \ref{sec:undecidable} investigates the model checking problem and satisfiability problem for SLL. We end this paper by Section \ref{sec:conclusion} on conclusion and future work.

\section{Modal Logic of Supervised Learning (SLL)}\label{sec:lml}
To be an ideal tool, the logic SLL should at least be able to precisely express players' actions and depict their winning strategies. In this section, we first introduce its language and semantics. Then we analyze its applications to SLG. Finally, we make various observations, including some logical validities and relations between SLL and other logics.

\subsection{Language and Semantics}\label{subsec:language and semantics}
We begin by considering the action of Learner. In SML, the standard modality $\Diamond$ characterizes the transition from a node to its successors and corresponds well to Learner's actions in SG. However, operator $\Diamond$ is not any longer sufficient in our case. Note that after Teacher cuts a link $\lr{w,v}$ from Learner's current process $S$, Learner should start from $w$ with the new path $S|_{\lr{w,v}}$ in the next round. Therefore, the desired operator should remember the history of Learner's movements.

To capture Teacher's action, a natural place to start is by defining operators that correspond to link addition and deletion. There is already a body of literature on logics of these modalities, such as the sabotage operator $\lr{-}$ and the bridge operator $\lr{+}$ \cite{changeoperator}. As mentioned, each occurrence of $\lr{-}$ in a formula deletes exactly one link whereas the bridge operator \textit{adds} links stepwise to models. Yet, including these two modalities is still not enough. For instance, we need to take into account whether or not a link deleted by Teacher is a part of the path of Learner's movements. We now introduce SLL. First, let us define its language.

\begin{definition}[Language $\mathcal{L}$]\label{def-language}
Let {\rm{\textbf{P}}} be a countable set of propositional atoms. The formulas of $\mathcal{L}$ are recursively defined in the following way:
$$\varphi::=p\mid\neg\varphi\mid(\varphi\land\varphi)\mid\blacklozenge\varphi\mid\lr{-}_1\varphi\mid\lr{-}_2\varphi\mid\lr{+}\varphi$$
where $p\in$\;{\rm{\textbf{P}}}. Notions $\top$, $\bot$, $\lor$ and $\to$ are as usual. Besides, we use $\blacksquare, [-]_1, [-]_2$ and $[+]$ to denote the dual operators of $\blacklozenge$, $\lr{-}_1$, $\lr{-}_2$ and $\lr{+}$ respectively.
\end{definition}

Several fragments of $\mathcal{L}$ will be studied in the following of the article. For brevity, we use a notational convention listing in subscript all modalities of the corresponding language. For instance, $\mathcal{L}_{\blacklozenge}$ is the fragment of $\mathcal{L}$ that has only the operator $\blacklozenge$ (besides Boolean connectives $\neg$ and $\land$); $\mathcal{L}_{\lr{-}_2}$ has only the modality $\lr{-}_2$; $\mathcal{L}_{\blacklozenge\lr{-}_1}$ has only the modality $\blacklozenge$ and $\lr{-}_1$, etc. We now proceed to define the models of SLL.

\begin{definition}[Models, Pointed Models and Frames]\label{def-model}
A model of SLL is a tuple $\mathcal{M}=\lr{W,R_1,R_2,V}$, where $W$ is a non-empty set of possible worlds, $R_{i\in\{1,2\}}\subseteq W^2$ are two binary relations and 
$V:\mathbf{P}\to2^W$ is a valuation function. $\mathcal{F}=\lr{W,R_1,R_2}$ is a frame. Let $S$ be an $R_1$-sequence, i.e., $Set(S)\subseteq R_1$. We name $\lr{\mathcal{M},S}$ a pointed model, and $S$ an evaluation sequence. Usually we write $\mathcal{M},S$ instead of $\lr{\mathcal{M},S}$.
\end{definition} 

For brevity, we call $R_1$ the \textit{conjectured relation} and $R_2$ the \textit{correct relation}. Besides, we use $\mathfrak{M}$ to denote the class of pointed models and $\mathfrak{M}^\bullet$ the class of pointed models whose sequence $S$ is a singleton. Before introducing the semantics, let us define some preliminary notations. 

Assume that $\mathcal{M}=\lr{W,R_1,R_2,V}$ is a model, $w\in W$ and $i\in\{1,2\}$. We use $R_i(w):=\{v\in W|R_iwv\}$ to denote the set of $R_i$-successors of $w$ in $\mathcal{M}$. Besides, for a sequence $S$, define $R_i(S):=R_i(e(S))$, i.e., the $R_i$-successors of a sequence $S$ are exactly the $R_i$-successors of its last element. For brevity, we also use $S;v$ to denote the sequence extending $S$ with node $v$. Moreover, $\mathcal{M}\ominus\lr{u,v}:=\lr{W,R_1\setminus\{\lr{u,v}\},R_2,V}$ is the model obtained by removing $\lr{u,v}$ from $R_1$, and $\mathcal{M}\oplus\lr{u,v}:=\lr{W,R_1\cup\{\lr{u,v}\},R_2,V}$ is obtained by extending $R_1$ in $\mathcal{M}$ with $\lr{u,v}$. We now have enough background to introduce the semantics of SLL.

\begin{definition}[Semantics]\label{def-semantics} 
Let $\lr{\mathcal{M},S}$ be a pointed model and $\varphi\in\mathcal{L}$. The semantics of SLL is defined as follows: 

\begin{center}
\begin{tabular}{|rcl|}
\hline
$\mathcal{M},S\vDash p$ & $\Leftrightarrow$ & $e(S)\in V(p)$ \\
$\mathcal{M},S\vDash\neg\varphi$ & $\Leftrightarrow$ & $ \mathcal{M},S \not\vDash\varphi$\\
$\mathcal{M},S\vDash\varphi\land\psi$ & $\Leftrightarrow$ & $\mathcal{M},S\vDash\varphi$ and $\mathcal{M},S\vDash\psi$\\
$\mathcal{M},S\vDash\blacklozenge\varphi$ & $\Leftrightarrow$ & $\exists v\in W$ s.t. $R_1e(S)v$ and $\mathcal{M},S;v\vDash\varphi$\\
$\mathcal{M},S\vDash\lr{-}_1\varphi$ & $\Leftrightarrow$ & $\exists\lr{v,v'}\in Set (S)\setminus R_2$ s.t. $\mathcal{M}\ominus\lr{v,v'},S|_{\lr{v,v'}}\vDash\varphi$ \\
$\mathcal{M},S\vDash\lr{-}_2\varphi$ & $\Leftrightarrow$ & $\exists\lr{v,v'}\in (R_1\setminus R_2)\setminus Set (S)$ s.t. $\mathcal{M}\ominus\lr{v,v'},S\vDash\varphi$ \\
$\mathcal{M},S\vDash\lr{+}\varphi$ & $\Leftrightarrow$ & $\exists\lr{v,v'}\in R_2\setminus R_1$ s.t. $\mathcal{M}\oplus\lr{v,v'},S\vDash\varphi$\\
\hline
\end{tabular}
\end{center}
\end{definition}

By the semantics, a propositional atom $p$ is true at a sequence $S$ if and only if $p$ is true at the last element of $S$. The cases for $\neg$ and $\land$ are as usual. Formula $\blacklozenge\varphi$ states that $e(S)$ has an $R_1$-successor $v$ such that $\varphi$ is true at sequence $S;v$. Additionally, $\lr{-}_1\varphi$ means that after deleting a link $\lr{v,v'}$ from $Set(S)\setminus R_2$, $\varphi$ is true at $S|_{\lr{v,v'}}$. Moreover, $\lr{-}_2\varphi$ states that $\varphi$ holds at $S$ after cutting a link $\lr{v,v'}$ belonging to $(R_1\setminus R_2)\setminus Set(S)$. Both $\lr{-}_1$ and $\lr{-}_2$ require that the link deleted cannot be an $R_2$-edge. Intuitively, whereas $\lr{-}_1$ depicts the case when Teacher deletes a link from Learner's path $S$, $\lr{-}_2$ captures the situation that the link deleted is not a part of $S$. Finally, $\lr{+}\varphi$ means that after extending $R_1$ with a new link belonging to $R_2$, $\varphi$ holds at the current sequence. 

A formula $\varphi$ is \textit{satisfiable} if there is a pointed model $\lr{\mathcal{M},S}\in\mathfrak{M}$ with $\mathcal{M},S\vDash\varphi$. \textit{Validity} in a model and in a frame is defined in the usual way. Note that the relevant class of models to specify SLL is $\mathfrak{M}^\bullet$, that is, models where the evaluation sequence $S$ starts with a singleton. Hence SLL is the set of $\mathcal{L}$-formulas that are valid in the class $\mathfrak{M}^\bullet$. 

For any $\lr{\mathcal{M},S}$ and $\lr{\mathcal{M}',S'}$, we say that they are \textit{learning modal equivalent} (notation: $\lr{\mathcal{M},S}\leftrightsquigarrow_l\lr{\mathcal{M}',S'}$) iff $\mathcal{M},S\vDash\varphi\Leftrightarrow\mathcal{M}',S'\vDash\varphi$ for any $\varphi\in\mathcal{L}$. The set $\mathbb{T}^l(\mathcal{M},S):=\{\varphi\in\mathcal{L}\mid\mathcal{M},S\vDash\varphi\}$ is the \textit{learning modal theory} of $\lr{\mathcal{M},S}$. Besides, we define a relation $\bf{U}\subseteq \mathfrak{M}\times\mathfrak{M}$ with $\lr{\lr{\mathcal{M},S},\lr{\mathcal{M}',S'}}\in \bf{U}$ iff $\lr{\mathcal{M}',S'}$ is $\lr{\mathcal{M}\ominus\lr{v,v'},S|_{\lr{v,v'}}}$ for some $\lr{v,v'}\in Set(S)\setminus R_2$, $\lr{\mathcal{M}\ominus\lr{v,v'},S}$ for some $\lr{v,v'}\in (R_1\setminus R_2)\setminus Set(S) $, or  $\lr{\mathcal{M}\oplus\lr{v,v'},S}$ for some $\lr{v,v'}\in R_2\setminus R_1$. We can also iterate this order, to talk about models reachable in finitely many $\bf{U}$-steps, obtaining the relation $\bf{U}^*$. 

\subsection{Application: Winning Strategies in SLG}\label{subsec:application}

By Definition \ref{def-semantics},  $\blacklozenge$ captures the actions of Learner, and operators $\lr{+}$, $\lr{-}_1$ and $\lr{-}_2$ characterize those of Teacher. Besides, our logic is expressive enough to describe the winning strategy (if there is one) for players in finite graphs.\footnote{Generally speaking, to define the existence of winning strategies for players, we need to extend SLG with some fixpoint operators. We leave this for future inquiry.}

Given a finite SLG, let $p$ be a distinguished atom holding only at the goal node. Generally, the winning strategy of Learner and Teacher will be of the following form:
\begin{align}
\blacksquare \bigcirc_{0} \blacksquare \bigcirc_{1} \blacksquare \cdots \bigcirc_{n} \blacksquare (p \wedge [-]_{1} \bot)
\end{align}
where $\bigcirc_{i}$ is blank or one of $\lr{-}_{1}$, $\lr{-}_{2}$ and $\lr{+}$ for each $i\le n$. In this formula, the recurring $\blacksquare$ operator depicts Learner's actions and $\bigcirc_{i}$ Teacher's response. The proposition $p$ signalizes Learner's arrival at the goal, and $[-]_{1} \bot$ states that there are no edges in Learner's path that Teacher can cut. Hence, we can conclude that Learner has reached the goal through a sequence of correct edges. It is worth noting that in formula (1) we use $\blacksquare$, other than $\blacklozenge$, to represent Learner's action, although SLG is a cooperative game. Recall the graph in Figure \ref{figure:clg}. We observe that $\blacksquare \lr{+}\blacksquare \lr{+} \blacksquare\lr{-}_1 \blacksquare\lr{-}_2\blacksquare\blacksquare(p\wedge [-]_{1} \bot)$ holds at the starting node $a$, so there exists a winning strategy in this specific SLG.

\begin{remark}
In SG we know that links cut by Teacher represent wrong inferences. However, SG does not tell us anything about the links that remain in the graph. Therefore, winning strategies of the players in SG cannot guarantee against situations like Gettier cases. In contrast, the formula $[-]_{1}\bot$ in (1) ensures that Teacher is not allowed to remove any more links from Learner's path. In SLG, a Gettier-style case is that Learner arrives at the goal node with some $\lr{u,v}\in R_1\setminus R_2$ occurring in her path, so Teacher now would be allowed to cut those links. Therefore Gettier cases cannot be winning strategies in SLG. 
\end{remark} 

\subsection{Preliminary Observations}\label{subsec:observations}
As observed, the semantics of SLL is not simple. In this section, we make some preliminary observations on SLL. In particular, we will discuss the relations between SLL and other related logics, and present some logical validities.  
 
First of all, we have the following result on the relation between $\mathcal{L}_{\blacklozenge}$ and standard modal logic:
 
\begin{proposition}\label{proposition:standard modal logic}
Let $\mathcal{M}=\lr{W,R_1,R_2,V}$ be a model. For any $\lr{\mathcal{M},S}\in\mathfrak{M}$ and $\varphi\in\mathcal{L}_{\blacklozenge}$, we have
$\mathcal{M},S\vDash\varphi\Leftrightarrow \lr{W,R_1,V},e(S)\vDash\varphi^*$, where $\varphi^*$ is a standard modal formula obtained by replacing every occurrence of $\blacklozenge$ in $\varphi$ with $\Diamond$.
\end{proposition}

\begin{proof}
The proof is done by induction on the syntax of $\varphi$. The Boolean cases are trivial. When $\varphi$ is $\blacklozenge\psi$, it holds that:
\begin{center}
\begin{tabular}{lll}
$\mathcal{M},S\vDash\varphi$&$\Leftrightarrow$& $\exists v\in R_1(e(S))$ s.t. $\mathcal{M},S;v\vDash\psi$\\
&$\Leftrightarrow$& $\exists v\in R_1(e(S))$ s.t. $\mathcal{M},v\vDash \psi^*$\\
&$\Leftrightarrow$& $\lr{W,R_1,V},e(S)\vDash \varphi^*$ 
\end{tabular}
\end{center}
The first equivalence follows from Definition \ref{def-semantics} directly. By the inductive hypothesis, the second one holds. The last one holds by the semantics of standard modal logic.
\qed
\end{proof}

Therefore, essentially the fragment $\mathcal{L}_{\blacklozenge}$ of $\mathcal{L}$ is standard modal logic. Moreover, the operator $\lr{-}_2$ is much similar to the sabotage operator $\lr{-}$:

\begin{proposition}\label{proposition:sabotage logic}
Let $\mathcal{M}=\lr{W,R_1,R_2,V}$ be a model, and $R=R_1\setminus R_2$. For any $\lr{\mathcal{M},w}\in\mathfrak{M}^\bullet$ and $\varphi\in\mathcal{L}_{\lr{-}_2}$, we have
$\mathcal{M},w\vDash\varphi\Leftrightarrow\lr{W,R,V},w\vDash\varphi'$, where $\varphi'$ is a SML formula obtained by replacing every occurrence of $\lr{-}_2$ in $\varphi$ with $\lr{-}$.
\end{proposition} 

\begin{proof}
We prove it by induction on the structure of $\varphi$. The Boolean cases are straightforward. When $\varphi$ is $\lr{-}_2\psi$, it holds that:
\begin{center}
\begin{tabular}{lll}
$\mathcal{M},w\vDash\varphi$ &$\Leftrightarrow$& $\exists \lr{v,v'}\in (R_1\setminus R_2)$ s.t. $\mathcal{M}\ominus\lr{v,v'},w\vDash\psi$\\
&$\Leftrightarrow$& $\exists \lr{v,v'}\in R$ s.t. $\lr{W,R\setminus\{\lr{v,v'}\},V},w\vDash \psi'$\\
&$\Leftrightarrow$& $\lr{W,R,V},w\vDash\lr{-}\psi'$ 
\end{tabular}
\end{center}
The first equivalence follows straightforward from the semantics of SLL. By the inductive hypothesis and the definition of $R$, we have the second equivalence. The last one holds by the truth condition for the sabotage modality.
\qed
\end{proof}

Next, we have the following result on the relation between $\mathcal{L}_{\blacklozenge\lr{+}}$ and the `\textit{bridge modal logic (BML)}' (i.e., the logic expanding the standard modal logic with the bridge operator):

\begin{proposition}\label{proposition:bridge logic}
Let $\mathcal{M}=\lr{W,R_1,W^2,V}$ be a model. For any $\lr{\mathcal{M},S}\in\mathfrak{M}$ and $\varphi\in\mathcal{L}_{\blacklozenge\lr{+}}$, we have $\mathcal{M},S\vDash\varphi\Leftrightarrow\lr{W,R_1,V},e(S)\vDash\varphi^\star$, where $\varphi^\star$ is a bridge modal formula obtained by replacing every occurrence of $\blacklozenge$ in $\varphi$ with $\Diamond$.\footnote{By abuse of notation, for any $\varphi\in\mathcal{L}_{\blacklozenge\lr{+}}$, $\varphi^\star$ is a formula of the bridge modal logic.}
\end{proposition} 

\begin{proof}
This goes by induction on the syntax of $\varphi$. The Boolean cases are trivial. The case for $\blacklozenge$ is similar to the proof of Proposition \ref{proposition:standard modal logic}. When $\varphi$ is $\lr{+}\psi$, it holds that:  
\begin{center}
\begin{tabular}{lll}
$\mathcal{M},S\vDash\varphi$
&$\Leftrightarrow$& $\exists \lr{v,v'}\in (R_2\setminus R_1)$ s.t. $\mathcal{M}\oplus\lr{v,v'},S\vDash\psi$\\
&$\Leftrightarrow$& $\exists v,v'\in W$ s.t. $\lr{v,v'}\not\in R_1$ and $\lr{W,R_1\cup\{\lr{v,v'}\},V},e(S)\vDash \psi'$\\
&$\Leftrightarrow$& $\lr{W,R_1,V},e(S)\vDash\lr{+}\psi'$ 
\end{tabular}
\end{center}
The first equivalence follows our semantics. By the inductive hypothesis and the definition of $R_2$, the second one holds. The last one holds by the semantics of bridge modal logic.
\qed
\end{proof}

From Proposition \ref{proposition:standard modal logic}-\ref{proposition:bridge logic}, we know that several fragments of SLL are similar to other logics that have been studied. However, as a whole, SLL is not a loose aggregation of these fragments: different operators interact with each other. A typical example is that, for any $\lr{\mathcal{M},w}\in \mathfrak{M}^\bullet$, the formula 
\begin{equation}
 [-]_1\varphi   
\end{equation}
is valid, as $Set(w)=\emptyset$. However, formula $\blacklozenge\neg[-]_1\varphi$ is satisfiable. This presents a drastic difference between SLL and other logics mentioned so far: in those logics, it is impossible that the evaluation point has access to a node satisfying a contradiction. In order to understand how operators in SLL work, we present some other validities of SLL.

\begin{proposition}
Let $p\in\bf{P}$ and $\varphi,\psi\in\mathcal{L}$. The following formulas are validities of SLL (w.r.t. $\mathfrak{M}^\bullet$):
\begin{align}
&&&p\to\blacksquare[-]_1p&&  \\
&&&p\land\blacklozenge\top\to\blacklozenge[-]_1p&&&\\
&&&p \to \bigcirc p&& \bigcirc\in\{[-]_2,[+]\} \\
&&&\bigcirc(\varphi\to\psi)\to(\bigcirc\varphi\to\bigcirc\psi) &&\bigcirc\in\{[-]_2,[+]\}\\
&&&\blacksquare^n[-]_1(\varphi\to\psi)\to(\blacksquare^n[-]_1\varphi\to\blacksquare^n[-]_1\psi)&& n\in N\\
&&&\blacksquare^n\lr{-}_1\varphi\to\blacksquare^{n+m}\lr{-}_1\varphi& & n,m\in N \\
&&&\blacklozenge^n\lr{-}_1\varphi\to\bigvee\limits_{m<n}\blacklozenge^m\lr{-}_2\varphi && 1\le n\in N
\end{align}
\end{proposition}

Note that formulas (3)-(5) above are not schemata. Although they will still be valid if we replace each propositional atom occurring in them with any Boolean formula, substitution fails in the general case. In particular, we have the following result:

\begin{proposition}\label{proposition:substitution}
Validities of $\mathcal{L}_{\blacklozenge\lr{-}_1}$ are not closed under substitution.
\end{proposition}

\begin{proof}
We prove it by example. Consider the general schema $\varphi\land\blacklozenge\psi\to\blacklozenge[-]_1\varphi$ of formula (4). Let $\varphi:=\blacklozenge p$ and $\psi:=\blacksquare q$. Define a model $\mathcal{M}$ as depicted in Figure \ref{figure:validity}. It holds that $\mathcal{M},w\vDash \blacklozenge p\land\blacklozenge\blacksquare q$. While, since $w$ has exactly one successor $w_1$ and $\lr{w,w_1}\not\in R_2$, we have $\mathcal{M},w\not\vDash\blacklozenge[-]_1\blacklozenge p$.\qed
\end{proof}

\begin{figure}
\centering
\begin{tikzpicture}
\node(a)[circle,draw,inner sep=0pt,minimum size=5mm] at (0,0) {$w$};
\node(b) [circle,draw,inner sep=0pt,minimum size=5mm][label=below:$p$] at (1.5,0){$w_1$};
\node(c)[circle,draw,inner sep=0pt,minimum size=5mm] [label=below:$q$] at (3,0){$w_2$};
\draw[->](a) to node [above] {$1$} (b);
\draw[->] (b) to node [above] {$2$} (c);
\end{tikzpicture}
\caption{A model of SLL}
\label{figure:validity}
\end{figure}
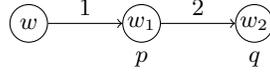

Interestingly, SLL also has other features that are very different from standard modal logic. For instance, it lacks the tree model property, which holds directly by the following result:

\begin{proposition}\label{proposition: tmp of fragement}
$\mathcal{L}_{\blacklozenge\lr{-}_1}$ does not have the tree model property.
\end{proposition}

\begin{proof}
Consider the following formulas: 
\begin{flalign*}
(T_1) \quad &  p\land\blacklozenge p\land\blacklozenge\neg p \\
(T_2) \quad & \blacksquare(p\to\blacklozenge p\land\blacklozenge\neg p) \\
(T_3) \quad & \blacksquare(\neg p\to\lr{-}_1(\blacksquare p\land\blacksquare\blacksquare p)) 
\end{flalign*}
Define $\varphi_T:=(T_1\land T_2\land T_3)$. We now show that, for any model $\mathcal{M}=\lr{W,R_1,R_2,V}$ and $w\in W$, if $\mathcal{M},w\vDash\varphi_T$, then $R_1ww$. By formula $(T_1)$, node $w$ is $p$, and it has at least one $p$-successor $w_1$ and at least one $\neg p$-successor $w_2$ via relation $R_1$. Besides, $(T_2)$ states that, each such $p$-successor $w_1$ of $w$ also can reach some $p$-node $w_3$ and $\neg p$-node $w_4$ by $R_1$. Finally, from $(T_3)$ we know that $w$ can only reach one $\neg p$-point by $R_1$ and that $w_1$ does not have $\neg p$-successors via $R_1$ any longer after cutting $\lr{w,w_2}$. So, $\lr{w,w_2}$ is identical with $\lr{w_1,w_4}$, which is followed by $R_1ww$ directly. 

Besides, formula $\varphi_T$ is indeed satisfiable with respect to $\mathfrak{M}^\bullet$. Consider the model depicted in Figure \ref{figure:reflexive}. It is not hard to see that $\varphi_T$ is true at $w$. Hence $\mathcal{L}_{\blacklozenge\lr{-}_1}$ lacks the tree model property.
\qed
\end{proof}

\begin{figure}
\centering
\begin{tikzpicture}
\node(a)[circle,draw,inner sep=0pt,minimum size=5mm] [label=below:$p$] at (0,0) {$w$};
\node(b)[circle,draw,inner sep=0pt,minimum size=5mm] at (2,0){$w_1$};
\draw[->](a) to  node[above] {$1$} (b);
\draw[->](a) to [in=160, out=200,looseness=8] node[left] {$1$}   (a);
\end{tikzpicture}
\caption{A model of $\varphi_T$.}
\label{figure:reflexive}
\end{figure}
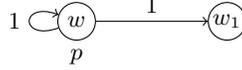

As observed, many instances of validities in our logic are not straightforward, and SLL has some distinguishing features. The results that we have so far are not sufficient enough to understand SLL. In the sections to come we will make a deeper investigation into our logic. 

\section{Expressive Power of SLL}\label{sec:expressivepower}

In this section, we study the expressivity of SLL. First, we will show that SLL is still a fragment of first-order logic even though it looks complicated. After this, a suitable notion of bisimulation for SLL is introduced. Finally, we provide a characterization theorem for the logic.

\subsection{First-Order Translation}
\label{subsec:firstorder}

Given the complicated semantics, is SLL still a fragment of FOL? In this section we will provide a positive answer to this question. To do so, we will describe a translation from SLL to FOL. However, compared with that for standard modal logic \cite{modallogic}, we now need some new devices. 

Let $\mathcal{L}_1$ be the first-order language consisting of countable unary predicates $P_{i\in N}$, two binary relations $R_{i\in\{1,2\}}$, and equality $\equiv$. Take any finite, non-empty sequence $E$ of variables. Let $y$ and $y'$ be two fresh variables not appearing in $E$. When there exists $\lr{x,x'}\in Set(E)$ with $x\equiv y$ and $x'\equiv y'$, we define $E|_{\lr{y,y'}}:=E|_{\lr{x,x'}}$. Now let us define the first-order translation.

\begin{definition}[First-Order Translation]\label{def-translation} Let $E=\lr{x_0,x_1,...,x_n}$ be a finite sequence of variables (non-empty), and $E^-=\{\lr{v_0,v'_0},...,\lr{v_i,v'_i}\}$ and $E^+=\{\lr{u_0,u'_0},...,\lr{u_j,u'_j}\}$ two finite sets of links. The translation $\mathcal{T}(\varphi,E,E^+,E^-)$ from $\mathcal{L}$-formulas $\varphi$ to first-order formulas is defined recursively as follows:
\begin{small}
\begin{align*}
\mathcal{T}(p,E,E^+,E^-)=&Pe(E)\\
\mathcal{T}(\neg\varphi,E,E^+,E^-)=&\neg\mathcal{T}(\varphi,E,E^+,E^-)\\
\mathcal{T}(\varphi\land\psi,E,E^+,E^-)=&\mathcal{T}(\varphi,E,E^+,E^-)\land \mathcal{T}(\psi,E,E^+,E^-)\\
\mathcal{T}(\blacklozenge\varphi,E,E^+,E^-)=&\exists y((\bigvee\limits_{\lr{x,x'}\in E^+}(e(E)\equiv x\land y\equiv x') \lor (R_1e(E)y\land\\
&\neg\bigvee\limits_{\lr{v,v'}\in E^-}(e(E)\equiv v\land y\equiv v')))\land\mathcal{T}(\varphi,E;y,E^+,E^-))\\
\mathcal{T}(\lr{-}_1\varphi,E,E^+,E^-)=&\exists y\exists y'(\bigvee\limits_{\lr{x,x'}\in Set(E)\setminus(E^-\cup E^+)}(y\equiv x\land y'\equiv x')\land\\
&R_1yy'\land\neg R_2yy'\land\mathcal{T}(\varphi,E|_{\lr{y,y'}},E^+,E^-\cup\{\lr{y,y'}\}))\\
\mathcal{T}(\lr{-}_2\varphi,E,E^+,E^-)=&\exists y\exists y'( R_1yy'\land\neg\bigvee\limits_{\lr{x,x'}\in Set(E)\cup E^-\cup E^+}(y\equiv x\land y'\equiv x')\land\\
&\neg R_2yy'\land\mathcal{T}(\varphi,E,E^+,E^-\cup\{\lr{y,y'}\}))\\
\mathcal{T}(\lr{+}\varphi,E,E^+,E^-)=&\exists y\exists y'(\neg\bigvee\limits_{\lr{x,x'}\in E^-\cup E^+}(y\equiv x\land y'\equiv x')\land\neg R_1yy'\land R_2yy'\land\\
&\mathcal{T}(\varphi,E,E^+\cup\{\lr{y,y'}\},E^-))
\end{align*}
\end{small}
\end{definition}

From the perspective of SLG, the sequence $E$ denotes Learner's process, and sets $E^+$ and $E^-$ represent the links that have already been added and deleted respectively. In any translation $\tau(\varphi,E,E^+,E^-)$, each of $E^+$ and $E^-$ may be extended. For any their extensions $E^+\cup X$ and $E^-\cup Y$, we have $X\cap Y=\emptyset$. Intuitively, this fact is in line with our semantics: for any $\lr{\lr{W,R_1,R_2,V},S}$, we always have $Set(S)\subseteq R_1$ and $(R_1\setminus R_2)\cap(R_2\setminus R_1)=\emptyset$, therefore links deleted are different from those added. Another point worth mentioning is that, unlike the case of standard modal logic, generally the translation does not yield a first-order formula with only one free variable. However, it does so when we set $E$, $E^+$ and $E^-$ to be a sequence consisting of a singleton, $\emptyset$ and $\emptyset$ respectively. By Definition \ref{def-translation}, we have the following result:

\begin{lemma}\label{lemma:translation}
Let $\mathcal{M}$ be a model and $\tau(\varphi,E,E^+,E^-)$ a translation s.t. $E^+\cap E^-=\emptyset$. Assume that $y$ and $y'$ are two fresh variables. For any assignment $\sigma$, we have $\mathcal{M}\ominus\lr{v,v'}\vDash\mathcal{T}(\varphi,E,E^+,E^-)[\sigma]$ iff $\mathcal{M}\vDash\mathcal{T}(\varphi,E,E^+,E^-\cup\{\lr{y,y'}\})[\sigma_{y^{(\prime)}:=v^{(\prime)}}]$, for any $\lr{v,v'}\in R_1\setminus R_2$; and $\mathcal{M}\oplus\lr{v,v'}\vDash\mathcal{T}(\varphi,E,E^+,E^-)[\sigma]$ iff $\mathcal{M}\vDash\mathcal{T}(\varphi,E,E^+\cup\{\lr{y,y'}\},E^-)[\sigma_{y^{(\prime)}:=v^{(\prime)}}]$, for any $\lr{v,v'}\in R_2\setminus R_1$.
\end{lemma}

\begin{proof}
The proofs for these two cases are similar. We only prove the first one. Assume that $\lr{v,v'}\in R_1\setminus R_2$. For brevity, define $R^-_1:=R_1\setminus\{\lr{v,v'}\}$, i.e., $R^-_1$ is the relation obtained by deleting the link $\lr{v,v'}$ from $R_1$ in $\mathcal{M}$.

(1). When $\varphi$ is $p\in \bf{P}$, we have the following equivalences:
\begin{center}
\begin{tabular}{ll}
&$\mathcal{M}\ominus\lr{v,v'}\vDash\mathcal{T}(\varphi,E,E^+,E^-)[\sigma]$ \\
$\Leftrightarrow$ \quad & $\mathcal{M}\ominus\lr{v,v'}\vDash Pe(E)[\sigma]$\\
$\Leftrightarrow$\quad & $\mathcal{M}\vDash Pe(E)[\sigma]$\\
$\Leftrightarrow$\quad & $\mathcal{M}\vDash\mathcal{T}(\varphi,E,E^+,E^-\cup\{\lr{y,y'}\})[\sigma_{y^{(\prime)}:=v^{(\prime)}}]$
\end{tabular}
\end{center}
The first equivalence holds by Definition \ref{def-translation} directly. The second one follows from the definition of $\mathcal{M}\ominus\lr{v,v'}$. The last one follows by Definition \ref{def-translation}.

(2). The proofs for the Boolean connectives are straightforward.

(3). $\varphi$ is $\blacklozenge\psi$. By $\lr{v,v'}\in R_1\setminus R_2$, the definitions of $\mathcal{M}\ominus\lr{v,v'}$ and the standard translation, it holds that:
\begin{center}
\begin{tabular}{ll}
&$\mathcal{M}\ominus\lr{v,v'}\vDash\mathcal{T}(\varphi,E,E^+,E^-)[\sigma]$ \\
$\Leftrightarrow$ \quad & $\mathcal{M}\ominus\lr{v,v'}\vDash\exists u(((\neg\bigvee\limits_{\lr{x,x'}\in E^-}(e(E)\equiv x\land u\equiv x')\land  R^-_1e(E)u)\lor$\\
&\qquad\qquad\qquad $\bigvee\limits_{\lr{z,z'}\in E^+}(e(E)\equiv z\land u\equiv z'))\land\mathcal{T}(\psi,E;u,E^+,E^-))[\sigma]$\\
$\Leftrightarrow$ \quad & $\mathcal{M}\vDash\exists u(( (\neg\bigvee\limits_{\lr{x,x'}\in E^-\cup\{y,y'\}}(e(E)\equiv x\land u\equiv x')\land  R_1e(E)u)\lor$\\
&\qquad\; $\bigvee\limits_{\lr{z,z'}\in E^+}(e(E)\equiv z\land u\equiv z'))\land$\\
&\qquad\;$\mathcal{T}(\psi,E;u,E^+,E^-\cup\{\lr{y,y'}\}))[\sigma_{y^{(\prime)}:=v^{(\prime)}}]$\\
$\Leftrightarrow$ \quad &$\mathcal{M}\vDash\mathcal{T}(\varphi,E,E^+,E^-\cup\{\lr{y,y'}\})[\sigma_{y^{(\prime)}:=v^{(\prime)}}]$
\end{tabular}    
\end{center}

(4). When $\varphi$ is $\lr{-}_1\psi$, we have that:
\begin{center}
\begin{tabular}{ll}
&$\mathcal{M}\ominus\lr{v,v'}\vDash\mathcal{T}(\varphi,E,E^+,E^-)[\sigma]$ \\
$\Leftrightarrow$ \quad & $\mathcal{M}\ominus\lr{v,v'}\vDash \exists u\exists u'(\bigvee\limits_{\lr{z,z'}\in Set(E)\setminus(E^+\cup E^-)}(u\equiv z\land u'\equiv z')\land R^-_1uu'\land $\\
&\qquad\qquad\qquad $\neg R_2uu'\land \mathcal{T}(\psi,E|_{\lr{u,u'}},E^+,E^-\cup\{\lr{u,u'}\}))[\sigma]$\\
$\Leftrightarrow$ \quad & $\mathcal{M}\vDash\exists u\exists u'(\bigvee\limits_{\lr{z,z'}\in Set(E)\setminus(E^+\cup (E^-\cup\{\lr{y,y'}\}))}(u\equiv z\land u'\equiv z')\land R_1uu'\land$\\
&\qquad\;$\neg R_2uu'\land\mathcal{T}(\psi,E|_{\lr{u,u'}},E^+,E^-\cup\{\lr{u,u'},\lr{y,y'}\}))[\sigma_{y^{(\prime)}:=v^{(\prime)}}]$\\
$\Leftrightarrow$ \quad &$\mathcal{M}\vDash\mathcal{T}(\varphi,E,E^+,E^-\cup\{\lr{y,y'}\})[\sigma_{y^{(\prime)}:=v^{(\prime)}}]$
\end{tabular}    
\end{center}

(5). $\varphi$ is $\lr{-}_2\psi$. The following equivalences hold:
\begin{center}
\begin{tabular}{ll}
&$\mathcal{M}\ominus\lr{v,v'}\vDash\mathcal{T}(\varphi,E,E^+,E^-)[\sigma]$ \\
$\Leftrightarrow$ &$\mathcal{M}\ominus\lr{v,v'}\vDash\exists u\exists u'( \neg\bigvee\limits_{\lr{z,z'}\in Set(E)\cup E^-\cup E^+}(u\equiv z\land u'\equiv z')\land$\\ &\qquad\qquad\qquad$\neg R_2uu'\land R^-_1uu'\land\mathcal{T}(\psi,E,E^+,E^-\cup\{\lr{u,u'}\}))[\sigma]$\\
$\Leftrightarrow$ &$\mathcal{M}\vDash\exists u\exists u'(\neg\bigvee\limits_{\lr{z,z'}\in Set(E)\cup E^+\cup (E^-\cup \{\lr{y,y'}\})}(u\equiv z\land u'\equiv z')\land$\\
&\qquad\;$ \neg R_2uu'\land R_1uu'\land\mathcal{T}(\psi,E,E^+,E^-\cup\{\lr{u,u'},\lr{y,y'}\}))[\sigma_{y^{(\prime)}:=v^{(\prime)}}]$\\
$\Leftrightarrow$ & $\mathcal{M}\vDash\mathcal{T}(\varphi,E,E^+,E^-\cup\{\lr{y,y'}\})[\sigma_{y^{(\prime)}:=v^{(\prime)}}]$
\end{tabular}    
\end{center}

(6). When $\varphi$ is $\lr{+}\psi$, it holds that:
\begin{center}
\begin{tabular}{ll}
&$\mathcal{M}\ominus\lr{v,v'}\vDash\mathcal{T}(\varphi,E,E^+,E^-)[\sigma]$ \\
$\Leftrightarrow$ & $\mathcal{M}\ominus\lr{v,v'}\vDash\exists u\exists u'(\neg\bigvee\limits_{\lr{z,z'}\in E^-\cup E^+}(u\equiv z\land u'\equiv z')\land\neg R^-_1uu'\land $\\
&\qquad\qquad\qquad $R_2uu'\land\mathcal{T}(\psi,E,E^+\cup\{\lr{u,u'}\},E^-))[\sigma]$\\
$\Leftrightarrow$ & $\mathcal{M}\vDash\exists u\exists u'(\neg\bigvee\limits_{\lr{z,z'}\in (E^-\cup\{y,y'\})\cup E^+}(u\equiv z\land u'\equiv z')\land\neg R_1uu'\land $\\
&\qquad\; $R_2uu'\land\mathcal{T}(\psi,E,E^+\cup\{\lr{u,u'}\},E^-\cup\{\lr{y,y'}\}))[\sigma_{y^{(\prime)}:=v^{(\prime)}}]$\\
$\Leftrightarrow$ & $\mathcal{M}\vDash\mathcal{T}(\varphi,E,E^+,E^-\cup\{\lr{y,y'}\})[\sigma_{y^{(\prime)}:=v^{(\prime)}}]$
\end{tabular}    
\end{center}
The proof is completed.\qed
\end{proof}

With Lemma \ref{lemma:translation}, we now can show the correctness of the translation:  

\begin{theorem}\label{theorem-correctnessoftranslation}
Let $\lr{\mathcal{M},S}$ be a pointed model and $E$ an $R_1$-sequence of variables with the same size as $S$. For any $\varphi\in\mathcal{L}$,
$\mathcal{M},S\vDash\varphi\;\;{\textit{iff}}\;\;\mathcal{M}\vDash\mathcal{T}\varphi, E, \emptyset,\emptyset)[E:=S].$
\end{theorem}

\begin{proof}
The proof is by induction on the structure of $\varphi$. 

(1). Formula $\varphi$ is a propositional atom $p\in\bf{P}$. By the semantics of SLL, $\mathcal{M},S\vDash\varphi$ if and only if $e(S)\in V(p)$. On the other hand, by Definition \ref{def-translation}, $\mathcal{T}(\varphi,E,\emptyset,\emptyset)$ is $Pe(E)$. So we have $\mathcal{M},S\vDash \varphi$ iff $\mathcal{M}\vDash\mathcal{T}(\varphi,E,\emptyset,\emptyset)[E:=S]$.

(2). The cases for Boolean connectives $\neg$ and $\land$ are trivial.

(3). When $\varphi$ is $\blacklozenge\psi$, we have the following equivalences:

\begin{center}
\begin{tabular}{lll}
&$\mathcal{M},S\vDash\varphi$ \\ 
$\Leftrightarrow$ \quad& $\exists v\in R_1(S)$ s.t. $\mathcal{M},S;v\vDash\psi$\\
$\Leftrightarrow$ & $\exists v\in R_1(S)$ s.t. $\mathcal{M}\vDash \mathcal{T}(\psi,E;y,\emptyset,\emptyset)[E:=S,y:=v]$\\
$\Leftrightarrow$ & $\mathcal{M}\vDash\exists y(R_1e(E)y\land\mathcal{T}(\varphi,E;y,\emptyset,\emptyset))[E:=S]$\\
$\Leftrightarrow$ & $\mathcal{M}\vDash \mathcal{T}(\varphi,E,\emptyset,\emptyset)[E:=S]$
\end{tabular}
\end{center}

(4). When $\varphi$ is $\lr{-}_1\psi$, it holds that:

\begin{center}
\begin{tabular}{lll}
&$\mathcal{M},S\vDash\varphi$\\
$\Leftrightarrow$\quad & $\exists\lr{v,v'}\in (Set(S)\setminus R_2)$ s.t. $\mathcal{M}\ominus\lr{v,v'},S|_{\lr{v,v'}}\vDash\psi$\\
$\Leftrightarrow$& $\exists\lr{v,v'}\in (Set(S)\setminus R_2)$ s.t.\\
&$\mathcal{M}\ominus\lr{v,v'}\vDash\mathcal{T}(\psi,E|_{\lr{y,y'}},\emptyset,\emptyset)[E:=S,y^{(\prime)}:=v^{(\prime)}]$\\
$\Leftrightarrow$& $\mathcal{M}\vDash\exists y\exists y'(\bigvee\limits_{\lr{v,v'}\in Set(E)}(y\equiv v\land y'\equiv v')\land R_1yy'\land\neg R_2yy'\land$\\
& \qquad \; $\mathcal{T}(\psi,E|_{\lr{y,y'}},\emptyset,\{\lr{y,y'}\}))[E:=S]$\\
$\Leftrightarrow$& $\mathcal{M}\vDash \mathcal{T}(\varphi,E,\emptyset,\emptyset)[E:=S]$
\end{tabular}
\end{center}

(5). If $\varphi$ is $\lr{-}_2\psi$, then we have: 

\begin{center}
\begin{tabular}{lll}
& $\mathcal{M},S\vDash\varphi$\\
$\Leftrightarrow$\quad & $\exists\lr{v,v'}\in (R_1\setminus R_2)\setminus Set(S)$ s.t. $\mathcal{M}\ominus\lr{v,v'},S\vDash\psi$\\
$\Leftrightarrow$& $\exists\lr{v,v'}\in (R_1\setminus R_2)\setminus Set(S)$ s.t. $\mathcal{M}\ominus\lr{v,v'}\vDash\mathcal{T}(\psi,E,\emptyset,\emptyset)[E:=S]$\\
$\Leftrightarrow$& $\mathcal{M}\vDash\exists y\exists y'(\neg\bigvee\limits_{\lr{v,v'}\in Set(E)}(y\equiv v\land y'\equiv v')\land R_1yy'\land\neg R_2yy'\land $\\
&\qquad\;$\mathcal{T}(\psi,E,\emptyset,\{\lr{y,y'}\}))[E:=S]$\\
$\Leftrightarrow$& $\mathcal{M}\vDash \mathcal{T}(\varphi,E,\emptyset,\emptyset)[E:=S]$
\end{tabular}   
\end{center}

(6). If $\varphi$ is $\lr{+}\psi$, then we have:

\begin{center}
\begin{tabular}{lll}
&$\mathcal{M},S\vDash\varphi$\\
$\Leftrightarrow$\quad & $\exists\lr{v,v'}\in R_2\setminus R_1$ s.t. $\mathcal{M}\oplus\lr{v,v'},S\vDash\psi$\\
$\Leftrightarrow$& $\exists\lr{v,v'}\in R_2\setminus R_1$ s.t.$\mathcal{M}\oplus\lr{v,v'}\vDash\mathcal{T}(\psi,E,\emptyset,\emptyset)[E:=S]$\\
$\Leftrightarrow$ & $\mathcal{M}\vDash\exists y\exists y'(R_2yy'\land\neg R_1yy'\land\mathcal{T}(\psi,E,\{\lr{y,y'}\},\emptyset))[E:=S]$\\
$\Leftrightarrow$&$\mathcal{M}\vDash \mathcal{T}(\varphi,E,\emptyset,\emptyset)[E:=S]$
\end{tabular}
\end{center}

This completes the proof.
\qed
\end{proof}

Note that the translation in Theorem \ref{theorem-correctnessoftranslation} has an extra requirement on the sequence $E$, i.e., $Set(E)\subseteq R_1$. Intuitively, this restriction corresponds to the definition of pointed models. For each $\lr{\mathcal{M},w}\in\mathfrak{M}^\bullet$, any extension $E'$ of $w$ fulfils the requirement naturally by Definition \ref{def-translation}.

\subsection{Bisimulation and Characterization for SLL}
\label{subsec:bisim}
The notion of bisimulation serves as a useful tool for establishing the expressive power of modal logics. However, it is not hard to see that SLL is not closed under the standard bisimulation \cite{modallogic}. In this section we introduce a novel notion of `learning bisimulation (l-bisimulation)' tailored to our logic, which leads to a characterization theorem for SLL as a fragment of first-order logic.

\begin{definition}[l-Bisimulation]\label{def-bisimulation} For any two models $\mathcal{M}=\lr{W,R_1,R_2,V}$ and $\mathcal{M}'=\lr{W',R'_1,R'_2,V'}$, a non-empty relation $Z_l \subseteq\bf{U}^*(\lr{\mathcal{M},S})\times\bf{U}^*(\lr{\mathcal{M}',S'})$ is an l-bisimulation between the two pointed models $\lr{\mathcal{M},S}$ and $\lr{\mathcal{M}',S'}$ (notation: $\lr{\mathcal{M},S}Z_l\lr{\mathcal{M}',S'}$) if the following conditions are satisfied:
\begin{enumerate}[align=left,itemindent=-1em]
\item[{\rm{\textbf{Atom:}}}] $\mathcal{M},S\vDash p$ iff $\mathcal{M}',S'\vDash p$, for each $p\in\textbf{\rm{\textbf{P}}}$.
\item[{\rm{\textbf{Zig$_{\blacklozenge}$:}}}] If there exists $v\in W_1$ s.t. $R_1wv$, then there exists $v'\in W_1$ s.t. $R_1'w'v'$ and $\lr{\mathcal{M},S;v}Z_l\lr{\mathcal{M}',S';v'}$.
\item[{\rm{\textbf{Zig$_{\lr{-}_1}$:}}}] If there is $\lr{u,v}\in Set(S)\setminus R_2$, then there is $\lr{u',v'}\in Set(S')\setminus R'_2$ with $\lr{\mathcal{M}\ominus\lr{u,v},S|_{\lr{u,v}}}Z_l\lr{\mathcal{M}'\ominus\lr{u',v'},S'|_{\lr{u',v'}}}$.
\item[{\rm{\textbf{Zig$_{\lr{-}_2}$:}}}] If there exists $\lr{u,v}\in (R_1\setminus R_2)\setminus Set(S)$, then there exists $\lr{u',v'}\in (R'_1\setminus R'_2)\setminus Set(S')$ with $\lr{\mathcal{M}\ominus\lr{u,v},S}Z_l\lr{\mathcal{M}'\ominus\lr{u',v'},S'}$.
\item[{\rm{\textbf{Zig$_{\lr{+}}$:}}}] If there exists $\lr{u,v}\in R_2\setminus R_1$, then there exists $\lr{u',v'}\in R'_2\setminus R'_1$ with $\lr{\mathcal{M}\oplus\lr{u,v},S}Z_l\lr{\mathcal{M}'\oplus\lr{u',v'},S'}$.
\item[{\rm{\textbf{Zag$_{\blacklozenge}$}}}, {\rm{\textbf{Zag$_{\lr{-}_1}$}}}, {\rm{\textbf{Zag$_{\lr{-}_2}$}}} and {\rm{\textbf{Zag$_{\lr{+}}$}}}:] the analogous clauses in the converse direction of {\rm{\textbf{Zig$_{\blacklozenge}$}}}, {\rm{\textbf{Zig$_{\lr{-}_1}$}}}, {\rm{\textbf{Zig$_{\lr{-}_2}$}}} and {\rm{\textbf{Zig$_{\lr{+}}$}}} respectively.
\end{enumerate}
For brevity, we write $\lr{\mathcal{M}_1,w}\underline{\leftrightarrow}_l\lr{\mathcal{M}_2,v}$ if there is an l-bisimulation $Z_l$ with $\lr{\mathcal{M}_1,w}Z_l\lr{\mathcal{M}_2,v}$.
\end{definition}

The clauses for $\blacklozenge$ is similar to those for the basic modality in the standard bisimulation: they keep the model fixed and extend the evaluation sequence with some of its $R_1$-successors. In contrast, all of the conditions for $\lr{-}_1$, $\lr{-}_2$ and $\lr{+}$ change the model. In particular, clauses for $\lr{-}_2$ and $\lr{+}$ do not modify the evaluation sequence, while those for $\lr{-}_1$ change both the model and the current sequence. Now we can show the following result:

\begin{theorem}[$\underline{\leftrightarrow}_l\subseteq\leftrightsquigarrow_l$]\label{theorem-bisimtoequiv}
For any pointed models $\lr{\mathcal{M},S}$ and $\lr{\mathcal{M}',S'}$, it holds that: $\lr{\mathcal{M},S}\underline{\leftrightarrow}_l\lr{\mathcal{M}',S'}\Rightarrow\lr{\mathcal{M},S}\leftrightsquigarrow_l\lr{\mathcal{M}',S'}$. 
\end{theorem}

\begin{proof}
The proof goes by induction on $\varphi$. Assume that $\lr{\mathcal{M},S}\underline{\leftrightarrow}_l\lr{\mathcal{M}',S'}$. The Boolean cases are straightforward.

(1). $\varphi$ is $\blacklozenge\psi$. If $\mathcal{M},S\vDash\varphi$, then there exists $v\in R_1(S)$ such that $\mathcal{M},S;v\vDash\psi$. By \textbf{Zig}$_{\blacklozenge}$, there exists $v'\in R'_1(S')$ such that $\lr{\mathcal{M},S;v}\underline{\leftrightarrow}_l\lr{\mathcal{M}',S';v'}$. By the inductive hypothesis, it holds that $\lr{\mathcal{M},S;v}\leftrightsquigarrow_l\lr{\mathcal{M}',S';v'}$, consequently, $\mathcal{M}',S';v'\vDash\psi$, which is followed by $\mathcal{M}',S'\vDash\varphi$ immediately. Similarly, we can obtain $\mathcal{M},S\vDash\varphi$ from $\mathcal{M}',S'\vDash\varphi$ by \textbf{Zag}$_{\blacklozenge}$.

(2). $\varphi$ is $\lr{-}_1\psi$. When $\mathcal{M},S\vDash\varphi$, there exists $\lr{u,v}\in Set(S)\setminus R_2$ such that $\mathcal{M}\ominus\lr{u,v},S|_{\lr{u,v}}\vDash\psi$. By {\rm{\textbf{Zig}$_{\lr{-}_1}$}}, there exists $\lr{u',v'}\in Set(S')\setminus R'_2$ with $\lr{\mathcal{M}\ominus\lr{u,v},S|_{\lr{u,v}}}\underline{\leftrightarrow}_l\lr{\mathcal{M}'\ominus\lr{u',v'},S'|_{\lr{u',v'}}}$. By the inductive hypothesis, $\lr{\mathcal{M}\ominus\lr{u,v},S|_{\lr{u,v}}}\leftrightsquigarrow_l\lr{\mathcal{M}'\ominus\lr{u',v'},S'|_{\lr{u',v'}}}$. So, $\mathcal{M}'\ominus\lr{u',v'},S'|_{\lr{u',v'}}\vDash\psi$, which is followed by $\mathcal{M}',S'\vDash\varphi$. In a similar way, when $\mathcal{M}',S'\vDash\varphi$, we can prove $\mathcal{M},S\vDash\varphi$ by {\rm{\textbf{Zag}$_{\lr{-}_1}$}}.

(3). $\varphi$ is $\lr{-}_2\psi$. If $\mathcal{M},S\vDash\varphi$, then there is $\lr{u,v}\in(R_1\setminus R_2)\setminus Set(S)$ with $\mathcal{M}\ominus\lr{u,v},S\vDash\psi$. By {\rm{\textbf{Zig}$_{\lr{-}_2}$}}, there exists $\lr{u',v'}\in (R'_1\setminus R'_2)\setminus Set(S')$ such that $\lr{\mathcal{M}\ominus\lr{u,v},S}\underline{\leftrightarrow}_l\lr{\mathcal{M}'\ominus\lr{u',v'},S'}$. By the inductive hypothesis, $\lr{\mathcal{M}\ominus\lr{u,v},S}\leftrightsquigarrow_l\lr{\mathcal{M}'\ominus\lr{u',v'},S'}$. Consequently, $\mathcal{M}'\ominus\lr{u',v'},S'\vDash\psi$. So we have $\mathcal{M}',S'\vDash\varphi$. Similarly, when $\mathcal{M}',S'\vDash\varphi$, we can prove $\mathcal{M},S\vDash\varphi$ by {\rm{\textbf{Zag}$_{\lr{-}_2}$}}.

(4). $\varphi$ is $\lr{+}\psi$. When $\mathcal{M},S\vDash\varphi$, there exists $\lr{u,v}\in R_2\setminus R_1$ such that $\mathcal{M}\oplus\lr{u,v},S\vDash\psi$. By \textbf{Zig}$_{\lr{+}}$, there exists $\lr{u',v'}\in R'_2\setminus R'_1$ with  $\lr{\mathcal{M}\oplus\lr{u,v},S}\underline{\leftrightarrow}_l\lr{\mathcal{M}'\oplus\lr{u',v'},S'}$. By IH, $\lr{\mathcal{M}\oplus\lr{u,v},S}\leftrightsquigarrow_l\lr{\mathcal{M}'\oplus\lr{u',v'},S'}$. Therefore we have $\mathcal{M}'\oplus\lr{u',v'},S'\vDash\psi$, consequently, $\mathcal{M}',S'\vDash\varphi$. Similarly, by \textbf{Zag}$_{\lr{+}}$, we know $\mathcal{M},S\vDash\varphi$ from $\mathcal{M}',S'\vDash\varphi$.
\qed
\end{proof}

Moreover, the converse direction of Theorem \ref{theorem-bisimtoequiv} holds for the models that are $\omega$-saturated. To introduce its definition, we need some auxiliary notations. For each finite set $Y$, we denote the expansion of $\mathcal{L}_{1}$ with a set $Y$ of constants with $\mathcal{L}_1^Y$, and denote the expansion of $\mathcal{M}$ to $\mathcal{L}_1^Y$ with $\mathcal{M}^Y$. Let $\mathbf{x}$ be a finite tuple of variables. 
A model $\mathcal{M}=\lr{W,R_1,R_2,V}$ is \textit{$\omega$-saturated} if, for every finite subset $Y$ of $W$, the expansion $\mathcal{M}^Y$ realizes every set $\Gamma(\mathbf{x})$ of $\mathcal{L}_1^Y$-formulas whose finite subsets $\Gamma'(\mathbf{x})$ are all realized in $\mathcal{M}^Y$.

\begin{theorem}[$\leftrightsquigarrow_l\subseteq\underline{\leftrightarrow}_l$]\label{theorem-omegaequivtobisim}
For any $\omega$-saturated $\lr{\mathcal{M},S}$ and $\lr{\mathcal{M}',S'}$, it holds that: $\lr{\mathcal{M},S}\leftrightsquigarrow_l\lr{\mathcal{M}',S'}\Rightarrow\lr{\mathcal{M},S}\underline{\leftrightarrow}_l\lr{\mathcal{M}',S'}$. 
\end{theorem}

\begin{proof}
We prove this by showing that $\leftrightsquigarrow_l$ itself is an l-bisimulation. In what follows, assume that $E'$ is an $R'_1$-sequence of variables with the same size as $S'$.

(1). For each $p\in\textbf{P}$, by the definition of $\leftrightsquigarrow_l$, it holds directly that $\mathcal{M},S\vDash p$ iff $\mathcal{M}',S'\vDash p$. This satisfies the condition of \textbf{Atom}. 

(2). Let $v\in R_1(S)$. We will prove that there is some $v'\in R'_1(S')$ with $\lr{\mathcal{M},S;v}\leftrightsquigarrow_l\lr{\mathcal{M}',S';v'}$. For any finite $\Gamma\subseteq\mathbb{T}^l(\mathcal{M},S;v)$, the following equivalences hold:

\begin{align*}
\mathcal{M},S\vDash\blacklozenge\bigwedge\Gamma &\;\;\Leftrightarrow\;\;  \mathcal{M}',S'\vDash\blacklozenge\bigwedge\Gamma\\
&\;\;\Leftrightarrow\;\;\mathcal{M}'\vDash\mathcal{T}(\blacklozenge\bigwedge\Gamma,E',\emptyset,\emptyset)[E':=S']\\
&\;\;\Leftrightarrow\;\;\mathcal{M}'\vDash\exists y(R'_1(E')y\land\mathcal{T}(\bigwedge\Gamma,E';y,\emptyset,\emptyset))[E':=S']
\end{align*}
The first equivalence holds by the assumption that $\lr{\mathcal{M},S}\leftrightsquigarrow_l\lr{\mathcal{M}',S'}$. The second one follows from Theorem \ref{theorem-correctnessoftranslation} and the last one from Definition \ref{def-translation}.

Since $\lr{\mathcal{M}',S'}$ is $\omega$-saturated, there exists some $y\in R'_1(E')$ such that $\mathcal{M}'\vDash\mathcal{T}(\mathbb{T}^l(\mathcal{M},S;v),E';y,\emptyset,\emptyset)[E':=S']$. Again by Theorem \ref{theorem-correctnessoftranslation}, there exists $v'\in R'_1(S')$ such that $\lr{\mathcal{M},S;v}\leftrightsquigarrow_l\lr{\mathcal{M}',S';v'}$. Now the proof of the {\rm{\textbf{Zig$_{\blacklozenge}$}}} clause is completed.

(3). Similar to (2), we can prove that the {\rm{\textbf{Zag$_{\blacklozenge}$}}} condition is satisfied.

(4). Let $\lr{u,v}\in Set(S)\setminus R_2$. We show that there exists $\lr{u',v'}\in Set(S')\setminus R'_2$ such that $\lr{\mathcal{M}\ominus\lr{u,v},S|_{\lr{u,v}}}\leftrightsquigarrow_l\lr{\mathcal{M}'\ominus\lr{u',v'},S'|_{\lr{u',v'}}}$. Let $\Gamma$ be a finite subset of $\mathbb{T}^l(\mathcal{M}\ominus\lr{u,v},S|_{\lr{u,v}})$, then we have the  following equivalences:
\begin{align*}
\mathcal{M},S\vDash\lr{-}_1\bigwedge\Gamma &\;\;\Leftrightarrow\;\; \mathcal{M}',S'\vDash\lr{-}_1\bigwedge\Gamma\\
&\;\;\Leftrightarrow\;\;\mathcal{M}'\vDash\mathcal{T}(\lr{-}_1\bigwedge\Gamma,E',\emptyset,\emptyset)[E':=S']\\
&\;\;\Leftrightarrow\;\;\mathcal{M}'\vDash\exists y\exists z(\bigvee\limits_{\lr{x,x'}\in Set(E')}(y\equiv x\land z\equiv x')\land\neg R'_2yz\land\\
&\qquad\qquad\quad\; \mathcal{T}(\bigwedge\Gamma,E'|_{\lr{y,z}},\emptyset,\{\lr{y,z}\}))[E':=S']
\end{align*}
The first equivalence holds straightforward from the  assumption of learning modal equivalence between $\lr{\mathcal{M},S}$ and $\lr{\mathcal{M}',S'}$. The second one follows from Theorem \ref{theorem-correctnessoftranslation}, and the third equivalence holds by Definition \ref{def-translation}.

Since $\lr{\mathcal{M}',S'}$ is $\omega$-saturated, there are $y,z$ such that $\lr{y,z}\in Set(E')\setminus R'_2$ and $\mathcal{M}'\vDash\mathcal{T}(\mathbb{T}^l(\mathcal{M}\ominus\lr{u,v},S|_{\lr{u,v}}),E'|_{\lr{y,z}},\emptyset,\{\lr{y,z}\}))[E':=S']$. W.l.o.g., assume that $y$ and $z$ are assigned to $u'$ and $v'$ respectively. By Lemma \ref{lemma:translation}, we have $\lr{u',v'}\in Set(S')\setminus R'_2$ with $\mathcal{M}'\ominus\lr{u',v'}\vDash\mathcal{T}(\mathbb{T}^l(\mathcal{M}\ominus\lr{u,v},S|_{\lr{u,v}}),E'|_{\lr{y,z}},\emptyset,\emptyset))[E':=S']$. By Theorem \ref{theorem-correctnessoftranslation}, $\mathcal{M}'\ominus\lr{u',v'},S'|_{\lr{u',v'}}\vDash\mathbb{T}^l(\mathcal{M}\ominus\lr{u,v},S|_{\lr{u,v}})$. Consequently, we have $\lr{\mathcal{M}\ominus\lr{u,v},S|_{\lr{u,v}}}\leftrightsquigarrow_l\lr{\mathcal{M}'\ominus\lr{u',v'},S'|_{\lr{u',v'}}}$. Therefore, the proof of the {\rm{\textbf{Zig$_{\lr{-}_1}$}}} clause is completed.

(5). Similar to (4), we can prove that the {\rm{\textbf{Zag$_{\lr{-}_1}$}}} condition is satisfied.

(6). Assume that $\lr{u,v}\in (R_1\setminus R_2)\setminus Set(S)$. We now prove that there exists $\lr{u',v'}\in (R'_1\setminus R'_2)\setminus Set(S')$ s.t. $\lr{\mathcal{M}\ominus\lr{u,v},S}\leftrightsquigarrow_l\lr{\mathcal{M}'\ominus\lr{u',v'},S'}$. For any finite $\Gamma\subseteq \mathbb{T}^l(\mathcal{M}\ominus\lr{u,v},S)$, it holds that:
\begin{align*}
\mathcal{M},S\vDash\lr{-}_2\bigwedge\Gamma &\;\;\Leftrightarrow\;\; \mathcal{M}',S'\vDash\lr{-}_2\bigwedge\Gamma\\
&\;\;\Leftrightarrow\;\;\mathcal{M}'\vDash\mathcal{T}(\lr{-}_2\bigwedge\Gamma,E',\emptyset,\emptyset)[E':=S']\\
&\;\;\Leftrightarrow\;\;\mathcal{M}'\vDash\exists y\exists z(\neg\bigvee\limits_{\lr{x,x'}\in Set(E')}(y\equiv x\land z\equiv x')\land R_1yy'\land\\
&\qquad\qquad\quad\;\neg R_2yy'\land \mathcal{T}(\bigwedge\Gamma,E',\emptyset,\{\lr{y,z}\}))[E':=S']
\end{align*}
The first equivalence follows from $\lr{\mathcal{M},S}\leftrightsquigarrow_l\lr{\mathcal{M}',S'}$ directly. The second holds from Theorem \ref{theorem-correctnessoftranslation}, and the third equivalence holds by Definition \ref{def-translation}.

As $\lr{\mathcal{M}',S'}$ is $\omega$-saturated, there are $y,z$ with $\lr{y,z}\in (R'_1\setminus R'_2)\setminus Set(E')$ and $\mathcal{M}'\vDash\mathcal{T}(\mathbb{T}^l(\mathcal{M}\ominus\lr{u,v},S),E',\emptyset,\{\lr{y,z}\}))[E':=S']$. W.l.o.g., assume that $y$ and $z$ are assigned to $u'$ and $v'$ respectively. From Lemma \ref{lemma:translation}, we know $\lr{u',v'}\in(R'_1\setminus R'_2)\setminus Set(S')$ and $\mathcal{M}'\ominus\lr{u',v'}\vDash\mathcal{T}(\mathbb{T}^l(\mathcal{M}\ominus\lr{u,v},S),E',\emptyset,\emptyset))[E':=S']$. By Theorem \ref{theorem-correctnessoftranslation},  it holds that $\mathcal{M}'\ominus\lr{u',v'},S'\vDash\mathbb{T}^l(\mathcal{M}\ominus\lr{u,v},S)$. So, we have $\lr{\mathcal{M}\ominus\lr{u,v},S}\leftrightsquigarrow_l\lr{\mathcal{M}'\ominus\lr{u',v'},S'}$. Now the proof of the {\rm{\textbf{Zig$_{\lr{-}_2}$}}} clause is completed.

(7). Similar to (6), we can show that {\rm{\textbf{Zag}$_{\lr{-}_2}$}} is also satisfied.

(8). Let $\lr{u,v}\in R_2\setminus R_1$. We now prove that the {\rm{\textbf{Zig}$_{\lr{+}}$}} condition is satisfied. Assume that $\Gamma$ is a finite subset of $\mathbb{T}^l(\mathcal{M}\oplus\lr{u,v},S)$. Then the following sequences hold:
\begin{align*}
\mathcal{M},S\vDash\lr{+}\bigwedge\Gamma&\Leftrightarrow \mathcal{M}',S'\vDash\lr{+}\bigwedge\Gamma\\
&\Leftrightarrow\mathcal{M}'\vDash\mathcal{T}(\lr{+}\bigwedge\Gamma,E',\emptyset,\emptyset)[E':=S']\\
&\Leftrightarrow\mathcal{M}'\vDash\exists y\exists z(R'_2yz\land\neg  R'_1yz\land\\
&\qquad\qquad\;\mathcal{T}(\bigwedge\Gamma,E',\{\lr{y,z}\},\emptyset))[E':=S']
\end{align*}
The first equivalence holds by $\lr{\mathcal{M},S}\leftrightsquigarrow_l\lr{\mathcal{M}',S'}$. The second one follows from Theorem \ref{theorem-correctnessoftranslation}, and the third equivalence holds by Definition \ref{def-translation}.

Note that $\lr{\mathcal{M}',S'}$ is $\omega$-saturated, hence there are $y,z$ such that $\lr{y,z}\in R'_2\setminus R'_1$ and $\mathcal{M}'\vDash\mathcal{T}(\mathbb{T}^l(\mathcal{M}\oplus\lr{u,v},S),E',\{\lr{y,z}\},\emptyset))[E':=S']$. W.l.o.g, assume that $y$ and $z$ are assigned to $u'$ and $v'$ respectively. By Lemma \ref{lemma:translation}, $\lr{u',v'}\in R'_2\setminus R'_1$ and $\mathcal{M}'\oplus\lr{u',v'}\vDash\mathcal{T}(\mathbb{T}^l(\mathcal{M}\oplus\lr{u,v},S),E',\emptyset,\emptyset))[E':=S']$. By Theorem \ref{theorem-correctnessoftranslation}, it follows that $\mathcal{M}'\oplus\lr{u',v'},S'\vDash\mathbb{T}^l(\mathcal{M}\oplus\lr{u,v},S)$. So, we have $\lr{\mathcal{M}\oplus\lr{u,v},S}\leftrightsquigarrow_l\lr{\mathcal{M}'\oplus\lr{u',v'},S'}$. Now the proof of the {\rm{\textbf{Zig$_{\lr{+}}$}}} clause is completed.

(9). Similar to (8), we can show the {\rm{\textbf{Zag}$_{\lr{+}}$}} condition is satisfied.
\qed
\end{proof}

Thus we have established a match between learning modal equivalence and learning bisimulation for the $\omega$-saturated models. Now, by a simple adaptation of standard arguments (cf. \cite{modallogic,sabotage}), we can show the following result:

\begin{theorem}\label{theorem-characterization}
For any $\alpha(x)\in\mathcal{L}_1$ with only one free variable, 
$\alpha(x)$ is equivalent to the translation of some $\mathcal{L}$-formula $\varphi$ iff $\alpha(x)$ is invariant under l-bisimulation.
\end{theorem}

\begin{proof}
The direction from left to right holds by Theorem \ref{theorem-bisimtoequiv} directly. We now consider the other direction. Let $\alpha$ be an $\mathcal{L}_1$-formula with only one free variable. Suppose that $\alpha$ is invariant under l-bisimulation. Define
$\mathbb{C}_l(\alpha):=\{\mathcal{T}(\varphi,x,\emptyset,\emptyset)\mid\varphi\in\mathcal{L}\ {\rm and}\ \alpha\vDash\mathcal{T}(\varphi,x,\emptyset,\emptyset)\}$. Note that any formula of $\mathbb{C}_l(\alpha)$ has only one free variable. We now show $\mathbb{C}_l(\alpha)\vDash\alpha$, i.e., $\mathcal{M}\vDash\mathbb{C}_l(\alpha)[x:=w]$ entails $\mathcal{M}\vDash\alpha[x:=w]$ for any $\lr{\mathcal{M},w}\in\mathfrak{M}^\bullet$. To do so, we first prove that the set $\Sigma=\mathcal{T}((\mathbb{T}^l(\mathcal{M},w),x,\emptyset,\emptyset))\cup\{\alpha\}$ is consistent.

Suppose that $\Sigma$ is not consistent. By the compactness of first-order logic, it holds that $\vDash\alpha\to\neg\bigwedge\Gamma$ for some finite $\Gamma\subseteq\mathcal{T}(\mathbb{T}^l(\mathcal{M},w),x,\emptyset,\emptyset)$. Then from the definition of $\mathbb{C}_l(\alpha)$, we know $\neg\bigwedge\Gamma\in\mathbb{C}_l(\alpha)$, which is followed by $\neg\bigwedge\Gamma\in \mathcal{T}(\mathbb{T}^l(\mathcal{M},w),x,\emptyset,\emptyset)$. However, it contradicts to $\Gamma\subseteq \mathcal{T}(\mathbb{T}^l(\mathcal{M},w),x,\emptyset,\emptyset)$. 

Now we show $\mathcal{M}\vDash\alpha[x:=w]$. Since $\Sigma$ is consistent, there exists some $\lr{\mathcal{M}',w'}\in\mathfrak{M}^\bullet$ s.t. $\mathcal{M}'\vDash\Sigma[x:=w]$. Consequently, $\lr{\mathcal{M},w}\leftrightsquigarrow_l\lr{\mathcal{M}',w'}$. Now take two $\omega$-saturated elementary extensions $\lr{\mathcal{M}_{\omega},w}$ and $\lr{\mathcal{M}'_{\omega},w'}$ of $\lr{\mathcal{M},w}$ and $\lr{\mathcal{M}',w'}$ respectively. It can be shown that such extensions always exist (see \cite{modeltheory}). By the invariance of first-order logic under elementary extensions, from $\mathcal{M}'\vDash\alpha[x:=w']$ we know $\mathcal{M}'_{\omega}\vDash\alpha[x:=w']$. Moreover, by Theorem \ref{theorem-omegaequivtobisim} and the assumption that $\alpha$ is invariant for l-bisimulation, we have $\mathcal{M}_{\omega}\vDash\alpha[x:=w]$. By the elementary extension, we obtain $\mathcal{M}\vDash\alpha[x:=w]$. Therefore, it holds that $\mathbb{C}_l(\alpha)\vDash\alpha$.

Finally, we show that $\alpha$ is equivalent to the translation of an $\mathcal{L}$-formula. Since $\mathbb{C}_l(\alpha)\vDash\alpha$, by the compactness and deduction theorems of first-order logic it holds that $\vDash\bigwedge\Gamma\to\alpha$ for some finite subset $\Gamma$ of $\mathbb{C}_l(\alpha)$. Besides, by the definition of $\mathbb{C}_l(\alpha)$, we have $\vDash\alpha\to\bigwedge\Gamma$. Thus, $\vDash\alpha\leftrightarrow\bigwedge\Gamma$. Now the proof is completed.
\qed
\end{proof}

Therefore, in terms of the expressivity, SLL is as powerful as the one free variable fragment of first-order logic that is invariant for l-bisimulation.

\section{Model Checking and Satisfiability for SLL}
\label{sec:undecidable}
In this section, we consider the the model checking problem and satisfiability problem for SLL. Fortunately, the results that we have already shown are quite helpful to establish the complexity result for its model checking problem. First of all, as noted in \cite{changeoperator}, it holds that:

\begin{theorem}\label{theorem: model checking bridge}
Model checking for BML is PSPACE-complete (see \cite{changeoperator}).
\end{theorem}

By this result, we now can show that model checking for SLL is also PSPACE-complete.

\begin{theorem}\label{theorem: model checking SLL}
Model checking for SLL is PSPACE-complete.
\end{theorem}

\begin{proof}
An upper bound can be established bt the first-order translation given by Definition \ref{def-translation}, which has only a polynomial size increase. It is well-known that model checking for FOL is in PSPACE.

Besides, an lower bound can be provided with the help of a translation $f$ from the bridge modal logic into the fragment $\mathcal{L}_{\blacklozenge\lr{+}}$. More precisely, $f$ is the reverse of the translation defined in Proposition \ref{proposition:bridge logic}. Clearly, the translation $f$ also has a polynomial size increase. Besides, let $\lr{W,R_1,V}$ be a standard relational model and $w\in W$. It is not hard to see that $\lr{W,R_1,V},w\vDash\varphi$ iff $\lr{W,R_1,W^2,V},w\vDash f(\varphi)$ (recall Proposition \ref{proposition:bridge logic}). From Theorem \ref{theorem: model checking bridge}, we know that model checking for SLL is PSPACE-hard.

Therefore, model checking for SLL is PSPACE-complete.
\qed
\end{proof}

Note that Theorem \ref{theorem: model checking SLL} also establishes an upper bound for the complexity of SLG. Now we move to considering the satisfiability problem for SLL. In particular, we have the following result:

\begin{theorem}\label{theorem: fmp of fragement}
$\mathcal{L}_{\blacklozenge\lr{-}_1}$ does not enjoy the finite model property.
\end{theorem}

\begin{proof}
To prove this, we present a formula that can only be satisfied by some infinite models. Consider the following formulas:
\begin{align*}
(F_1) &&& p\land q\land\blacklozenge p\land\blacklozenge\neg p\land\blacksquare\neg q \\  
(F_2)&&& \blacksquare(p\to\blacklozenge q\land\blacklozenge\neg q\land\blacksquare p)\\
(F_3)&&& \blacksquare(p\to\blacksquare(q\to\blacksquare\neg q\land \blacklozenge\neg p))\\
(F_4)&&&\blacklozenge(\neg p\land \lr{-}_1\blacksquare(p\land\blacksquare(q\to\blacksquare p)))\\
(F_5)&&&\blacksquare(p\to\blacksquare(\neg q\to\blacklozenge q\land\blacklozenge\neg q\land\blacksquare p))\\
(F_6)&&& \blacksquare(p\to\blacksquare(\neg q\to\blacksquare(q\to\blacksquare\neg q\land\blacklozenge\neg p)))\\
(F_7)&&& \blacklozenge(\neg p\land\lr{-}_1\blacksquare\blacksquare(\neg q\to\blacksquare(q\to\blacksquare p)))\\
(\textit{Spy})&&& \blacksquare(p\to\blacksquare(\neg q\to\blacksquare(q\to\lr{-}_1(\neg q\land\blacksquare\neg q\land\lr{-}_1(q\land\blacklozenge(p\land\blacksquare\neg q))))))\\
(\textit{Irr})&&& \blacksquare(p\to\blacksquare(q\to\lr{-}_1(\neg q\land\blacksquare\neg q\land\blacksquare\blacklozenge q)))\\
(\textit{No-3cyc})&&& \neg\blacklozenge(p\land\blacksquare(q\to\lr{-}_1(\neg q\land\blacksquare(\neg q\land\blacklozenge\blacklozenge(p\land\blacksquare\neg q)))\\
(\textit{Trans})&&&\blacksquare(p\to\blacksquare(q\to\lr{-}_{1}(\neg q\land\blacksquare\neg q\land\blacksquare\blacksquare(\neg q\to\blacksquare(q\to\lr{-}_1(\neg q\land \blacksquare\neg q\land\\
&&&\lr{-}_1(p\land\neg\blacklozenge q\land\blacklozenge\blacksquare q)))))))
\end{align*}

Let formula $\varphi_\infty$ be the conjunction of the formulas above. We first show that $\varphi_\infty$ is satisfiable. Consider the model depicted in Figure \ref{figure:infinite}. It holds that $\varphi_\infty$ is true at $w$.

Next, we prove that for any model $\mathcal{M}=\{W,R_1,R_2,V\}$ and $w\in W$, if $\mathcal{M},w\vDash\varphi_\infty$, then $W$ is infinite. For brevity, define that $B=\{v\in W\mid v\in R_1(w)\cap V(p)\}$, i.e., $B$ is the set of the $p$-points that can be reached by $w$ in one step via $R_1$. In what follows, we assume that all previous conjuncts hold.

By $(F_1)$, node $w$ is $(p\land q)$, and it cannon see any $q$-points via $R_1$. In particular, it cannot see itself via $R_1$. Besides, $w$ has at least one $p$-successor and at least one $\neg p$-successor via $R_1$, i.e., $B\not=\emptyset$ and $R_1(w)\setminus B\not=\emptyset$.

From formula $(F_2)$, we know that each element in $B$ can see some $(q\land p)$-point(s) and $(\neg q\land p)$-point(s) via $R_1$, but cannot see any $\neg p$-points by $R_1$. Hence each point in $B$ has at least one $R_1$-successor distinct from itself. 

According to formula $(F_3)$, for any $w_1\in B$, each its $R_1$-successor that is $q$ can see some $\neg p$-point(s) via $R_1$, but cannot see any $q$-points by $R_1$.

By $(F_4)$, it holds that $R_1(w)\setminus B\not=\emptyset$ includes only one element. Moreover, each $w_1\in B$ can see point $w$ via $R_1$, and for each $q$-point $w_2\in W$, if $w_2$ is a successor of $w_1$ via $R_1$, then $w_2$ must be $w$.  

Formulas $(F_2)$-$(F_4)$ show the properties of the $(\neg q\land p)$-points which are accessible from the point $w$ in one step by $R_1$. Similarly, formulas $(F_5)$, $(F_6)$ and $(F_7)$ play the same role as $(F_2)$, $(F_3)$ and $(F_4)$ respectively, but focusing on showing the properties of the $(\neg q\land p)$-points that are accessible from $w$ in two steps via $R_1$. In particular, $(F_7)$ guarantees that every $(\neg q\land p)$-point $w_1$ which is accessible from $w$ in two steps by $R_1$ can also see $w$ via $R_1$, and that for each $q$-point $w_2\in W$, if $R_1w_1w_2$, then $w_2$ must be $w$.

Formula (\textit{Spy}) is a bit complicated. It shows that, for any two $(\neg q\land p)$-points $w_1$ and $w_2$ such that $R_1ww_1$ and $R_1w_1w_2$, after we delete some link $\lr{v,v'}\in\{\lr{w,w_1},\lr{w_1,w_2},\lr{w_2,w}\}$, $v$ is $\neg q$ and does not have any $q$-successors. Since $w$ is $q$, $v$ cannot be $w$. Besides, if $\lr{v,v'}=\lr{w_1,w_2}$, after we cut the link $\lr{v,v'}$, $v$ still have one $q$ successor, i.e., $w$, so we have $\lr{v,v'}=\lr{w_2,w_w}$. Further more, after we delete $\lr{w,w_1}$, $w$ can reach a $p$-point $w_3$ via $R_1$ such that $w_3$ has no $q$ successor via $R_1$. Therefore, $w_3$ must be $w_2$. In such a way, (\textit{Spy}) ensures that each $(\neg q\land p)$-point $w_1$ which is accessible from $w$ in two steps via $R_1$ is also accessible from $w$ in one step via $R_1$.

By (\textit{Irr}), each $w_1\in B$ is irreflexive. Finally, (\textit{No-3cyc}) shows that the accessibility relations of $R_1$ cannot be cycles of length 2 or 3 in $B$, and (\textit{Trans}) forces the accessibility relation $R_1$ to transitively order $B$.

Hence $(B,R_1)$ is an unbounded strict partial order, thus $B$ is infinite and so is $W$. Now we know that $\varphi_\infty$ is satisfiable, and that for each $\lr{\mathcal{M},w}$, if $\mathcal{M},w\vDash\varphi_\infty$, then $\mathcal{M}$ is an infinite model. This completes the proof.\qed
\end{proof}

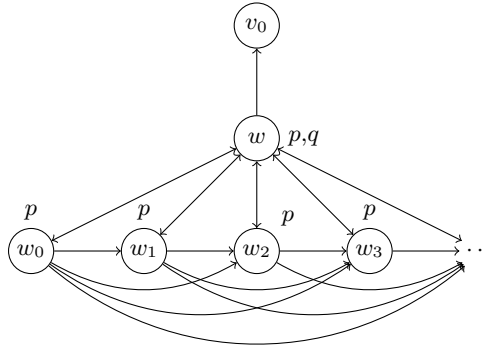
\begin{figure}
\centering
\begin{tikzpicture}
\node(a)[circle,draw,inner sep=0pt,minimum size=6mm, label=right:{$p$,$q$}] at (0,0) {$w$};
\node(c)[circle,draw,inner sep=0pt,minimum size=6mm] at (0,1.5) {$v_0$};
\node(e)[circle,draw,inner sep=0pt,minimum size=6mm,label=above:$p$] at (-3,-1.5) {$w_0$};
\node(f)[circle,draw,inner sep=0pt,minimum size=6mm,label=above:$p$] at (-1.5,-1.5) {$w_1$};
\node(g)[circle,draw,inner sep=0pt,minimum size=6mm,label=above right:$p$] at (0,-1.5) {$w_2$};
\node(h)[circle,draw,inner sep=0pt,minimum size=6mm,label=above:{$p$}] at (1.5,-1.5) {$w_3$};
\node(i)[circle,inner sep=0pt,minimum size=6mm] at (3,-1.5) {$\cdots$};
\draw[->](a) to (c);
\draw[<->](a) to (e);
\draw[<->](a) to (f);
\draw[<->](a) to (g);
\draw[<->](a) to (h);
\draw[<->](a) to (i);
\draw[->](e) to (f);
\draw[->](f) to (g);
\draw[->](g) to (h);
\draw[->](h) to (i);
\draw[->](e) to[bend right=30] (g);
\draw[->](f) to[bend right=30] (h);
\draw[->](f) to[bend right=35] (i);
\draw[->](e) to[bend right=35] (h);
\draw[->](e) to[bend right=40] (i);
\draw[->](g) to[bend right=30] (i);
\end{tikzpicture}
\caption{A model of formula $\varphi_\infty$ (every link in the model belongs to $R_1$, and $R_2=\emptyset$).}
\label{figure:infinite}
\end{figure}

We now proceed to show the undecidability of $\mathcal{L}_{\blacklozenge\lr{-}_1}$. To do so, we will reduce the $N\times N$ tiling problem to the satisfiability problem of this fragment. 

A tile $t$ is a $1\times1$ square, of fixed orientation, with colored edges \textit{right($t$)}, \textit{left($t$)}, \textit{up($t$)} and \textit{down($t$)}. The $N\times N$ tiling problem is: given a finite set of tile types $T$, is there a function $f:N\times N\to T$ with \textit{right(f(n,m))}=\textit{left(f(n+1,m))} and \textit{up(f(n,m))}=\textit{down(f(n,m+1))}? This problem is known to be undecidable (\cite{tile}).

Inspired by the technique in \cite{hybrid}, we will use three kinds of modalities $\blacklozenge_s$, $\blacklozenge_u$ and $\blacklozenge_r$ to stand for $\blacklozenge$. Correspondingly, a model $\mathcal{M}=\{W,R_s,R_u,R_r,R_2,V\}$ now has four kinds of relations. We are going to construct a spy point over relation $R_s$. Besides, $R_u$ and $R_r$ represent moving up and to the right, respectively, from one tile to the other. Intuitively, the union of these three relations can be treated as the $R_1$ relation of the model. Moreover, as illustrated by the following proof, they are disjoint with each other. So they are a partition of $R_1$. Thanks to this fact, we do not need any extra modalities to represent $\lr{-}_1$.

\begin{theorem}\label{theorem:undecidability of fragement}
The satisfiability problem for $\mathcal{L}_{\blacklozenge\lr{-}_1}$ is undecidable.
\end{theorem}

\begin{proof}
Assume that $T=\{T_1,...,T_n\}$ be a finite set of tile types. For each $T_i\in T$, $u(T_i)$, $d(T_i)$, $l(T_i)$ and $r(T_i)$ are the colors of its up, down, left and right edges respectively. Besides, each tile type is coded with a fixed propositional atom $t_i$. Now we will show that $\varphi_T$, the conjunction of the following formulas, is satisfiable iff $T$ tiles $N\times N$. 

\begin{align*}
(M_1)&&& p\land q\land \blacklozenge_sp\land\blacklozenge_s\neg p\land \blacksquare_s\neg q\land\blacklozenge_s\lr{-}_1\blacksquare_sp\\
(M_2)&&& \blacksquare_s(p\to\blacklozenge_s\top\land\blacksquare_s(q\land\blacklozenge_s\neg p))\\
(M_3)&&& \blacklozenge_s(\neg p\land\lr{-}_1\blacksquare_s\blacksquare_s(q\land\neg\blacklozenge_s\neg p))\\
(M_4)&&&\blacksquare_s(p\to\blacklozenge_u\top\land\blacksquare_u(p\land\neg q\land\blacklozenge_s\top\land\blacksquare_s(q\land \blacklozenge_s\neg p)))\\
&&&\blacksquare_s(p\to\blacklozenge_r\top\land\blacksquare_r(p\land\neg q\land\blacklozenge_s\top\land\blacksquare_s(q\land \blacklozenge_s\neg p)))\\
(M_5)&&&\blacklozenge_s(\neg p\land\lr{-}_1\blacksquare_s\blacksquare_u\blacksquare_s\neg\blacklozenge_s\neg p)\\
&&&\blacklozenge_s(\neg p\land\lr{-}_1\blacksquare_s\blacksquare_r\blacksquare_s\neg\blacklozenge_s\neg p)\\
(M_6)&&& \blacksquare_s(p\to\blacksquare_u(\blacklozenge_u\top\land\blacklozenge_r\top\land\blacksquare_u(p\land\neg q)\land\blacksquare_r(p\land\neg q)))\\
&&& \blacksquare_s(p\to\blacksquare_r(\blacklozenge_u\top\land\blacklozenge_r\top\land\blacksquare_u(p\land\neg q)\land\blacksquare_r(p\land\neg q)))\\
(M_7)&&& \blacksquare_s(p\to\blacksquare_s(q\land\lr{-}_1(\neg q\land\blacksquare_u(\blacklozenge_sq\land\neg\blacklozenge_u\neg\blacklozenge_sq))))\\
&&& \blacksquare_s(p\to\blacksquare_s(q\land\lr{-}_1(\neg q\land\blacksquare_r(\blacklozenge_sq\land\neg\blacklozenge_r\neg\blacklozenge_sq))))\\
(\textit{Spy})&&& \blacksquare_s(p\to\blacksquare_u\blacksquare_s\lr{-}_1(\blacksquare_s\bot\land\lr{-}_1(p\land q\land\blacklozenge_s(p\land\blacksquare_s\bot))))\\
&&& \blacksquare_s(p\to\blacksquare_r\blacksquare_s\lr{-}_1(\blacksquare_s\bot\land\lr{-}_1(p\land q\land\blacklozenge_s(p\land\blacksquare_s\bot))))\\
(\textit{Func})&&& \blacksquare_s(p\to\blacksquare_s\lr{-}_1(\blacksquare_s\bot\land\blacksquare_u\lr{-}_1(\blacksquare_s\bot\land\blacksquare_u\bot))\\
&&& \blacksquare_s(p\to\blacksquare_s\lr{-}_1(\blacksquare_s\bot\land\blacksquare_r\lr{-}_1(\blacksquare_s\bot\land\blacksquare_r\bot))\\
(\textit{No-UR})&&&\blacksquare_s(p\to\blacksquare_s\lr{-}_1(\blacksquare_s\bot\land\blacksquare_u\blacksquare_r\blacklozenge_sq\land\blacksquare_r\blacksquare_u\blacklozenge_sq)) \\
(\textit{No-URU})&&& \blacksquare_s(p\to\blacksquare_s\lr{-}_1(\blacksquare_s\bot\land\blacksquare_u\blacksquare_r\blacksquare_u\blacklozenge_sq)) \\
(\textit{Conv})&&& \blacksquare_s(p\to\blacksquare_s\lr{-}_1(\blacksquare_s\bot\land\blacklozenge_u\blacksquare_s\lr{-}_1(\blacksquare_s\bot\land\blacklozenge_u\top\land \\
&&&\blacklozenge_r\blacksquare_u\lr{-}_1(\blacksquare_u\bot\land \blacklozenge_s\blacklozenge_s(p\land\blacksquare_s\bot\land\blacklozenge_r\blacklozenge_u\top\land\blacklozenge_r\blacklozenge_u(p\land\blacksquare_u\bot)))))) \\
(\textit{Unique})&&&\blacksquare_s(p\to\bigvee\limits_{1\le i\le n}t_i\land\bigwedge\limits_{1\le i<j\le n}(t_i\to\neg t_j))\\
(\textit{Vert})&&& \blacksquare_s(p\to \bigwedge\limits_{1\le i\le n}(t_i\to\blacklozenge_u\bigvee\limits_{1\le j\le n,\;u(T_i)=d(T_j)}t_j))\\
(\textit{Horiz})&&&\blacksquare_s(p\to \bigwedge\limits_{1\le i\le n}(t_i\to\blacklozenge_r\bigvee\limits_{1\le j\le n,\;r(T_i)=l(T_j)}t_j))
\end{align*}

Let $\mathcal{M}=\{W,R_s,R_u,R_r,R_2,V\}$ be a model and $w\in W$ such that $\mathcal{M},w\vDash\varphi_T$. We now show that $\mathcal{M}$ is a tiling of $N\times N$. Define $G:=\{v\in W\mid v\in R_s(w)\cap V(p)\}$ where $R_s(w)=\{v\in W\mid R_swv\}$, and we will use its elements to represent the tiles.

By $(M_1)$, node $w$ is $(p\land q)$, and it cannon see any $q$-points via $R_s$. So, $\neg R_sww$. Besides, $w$ has exactly one $\neg p$-successor (e.g., v) and some $p$-successor(s) via $R_1$, i.e., $G\not=\emptyset$ and $R_s(w)\setminus G=\{v\}$.

By $(M_2)$, each tile $w_1$ has some successor(s) via $R_s$, and each such successor $w_2$ is $q$ and also has some $\neg p$-successor(s) via $R_s$. It is worth noting that $(M_1)$ and $(M_2)$ illustrate that $R_s$ is irreflexive.

Formula $(M_3)$ ensures that each tile $w_1$ can see $w$ via $R_s$, and that for each $(q\land p)$-point $w_2\in W$, if $w_2$ is accessible from $w_1$ via $R_s$, then $w_2=w$.

From $(M_4)$, we know that each tile has some successor(s) via $R_u$ and some successor(s) via $R_r$. Besides, each point that is accessible from a tile via $R_u$ or $R_r$ is $(\neg q\land p)$, and it has some $q$-successor(s) $w_1$ via relation $R_s$ where each $w_1$ can see some $\neg p$-point(s) via $R_s$.  

By formula $(M_5)$, each $w_1\in W$ accessible from a tile via $R_u$ or $R_r$ can see $w$ by $R_s$. Also, for each $(q\land p)$-point $w_2\in W$, if it is accessible from $w_1$ via $R_s$, then $w_2=w$.

Formula $(M_6)$ ensures that each $w_1\in W$ that is accessible from a tile via $R_u$ or $R_r$ also has some successor(s) via $R_u$ and some successor(s) via $R_r$. Besides, each its successor via $R_u$ or $R_r$ is $(\neg q\land p)$.

From formula $(M_7)$, it follows that both $R_u$ and $R_r$ are irreflexive and asymmetric.

By (\textit{Spy}), $w$ is a spy point via the relation $R_s$.

Note that formula $(M_4)$ says that each tile has some tile(s) above it and some tile(s) to its right. Now, with (\textit{Func}), we have that each tile has exactly one tile above it and exactly one tile to its right.

By (\textit{No-UR}), any tile cannot be above/below as well as to the left/right of another tile. Formula (\textit{No-URU}) disallows cycles following successive steps of the $R_u$, $R_r$, and $R_u$ relations, in this order. Moreover, (\textit{Conv}) ensures that the tiles are arranged as a grid.

Formula (\textit{Unique}) guarantees that each tile has a unique type. Finally, (\textit{Vert}) and (\textit{Horiz}) force the colors of the tiles to match properly.

Thus we conclude that $\mathcal{M}$ is indeed a tiling of $N\times N$. 

\smallskip

Next we show the other direction required for our proof. Suppose the function $f: N\times N\to T$ is a tiling of $N\times N$. Define a model $\mathcal{M}=\{W,R_s,R_u,R_r,R_2,V\}$ as follows:
\begin{align*}
W&=(N\times N)\cup\{w,v\}\\
R_s&=\{\lr{w,v}\}\cup\{\lr{w,x}\mid x\in N\times N\}\cup\{\lr{x,w}\mid x\in N\times N\}\\
R_u&=\{\lr{\lr{n,m},\lr{n,m+1}}\mid n,m\in N\}\\ 
R_r&=\{\lr{\lr{n,m},\lr{n+1,m}}\mid n,m\in N\}\\ 
R_2&=\emptyset\\
V(q)&=\{w\}\\ 
V(p)&=\{w\}\cup (N\times N)\\
V(t_i)&=\{\lr{n,m}\in N\times N\mid f(\lr{n,m})=T_i\},\; {\rm{for\; each }}\; i\in\{1,...,n\}\\
V(r)&=\emptyset,\; {\rm{for\; any\; other\; propositional\; atoms}}\; r
\end{align*}
In particular, $w$ is a spy point in $\mathcal{M}$. By construction, we have $\mathcal{M},w\vDash\varphi_T$.
\qed
\end{proof}

By Theorem \ref{theorem: fmp of fragement}-\ref{theorem:undecidability of fragement}, it holds directly that:

\begin{theorem}\label{theorem:undecidability of sll}
SLL lacks the finite model property, and its satisfiability problem is undecidable. 
\end{theorem}

\section{Conclusion and Future Work}
\label{sec:conclusion}

\textbf{Summary} Motivated by restrictions on learning in SG, we have extended the game to SLG by naming right and wrong paths of learning, and let Teacher not only delete but also add links. Afterwards, logic SLL was presented, which enables us to reason about players' strategies in SLG. Besides, to understand the new device, we provided some interesting observations and logical validities. Next, we studied basics of its expressivity, including its first-order translation, a novel notion of bisimulation and a characterization theorem for SLL as a fragment of FOL that is invariant under the bisimulation introduced. Finally, it was proved that model checking for SLL is PSPACE-complete, and via the research on $\mathcal{L}_{\blacklozenge\lr{-}_1}$ we shown that SLL does not enjoy the finite model property and its satisfiability problem is undecidable. 

\vspace{.3cm}

\noindent\textbf{Relevant and Future Research} Broadly, this work takes a small step towards studying the interaction between graph games, logics and formal learning theory. We are inspired
by the work on SG \cite{lig}, SML \cite{sabotage} and their application to formal learning theory \cite{learning}. This article is also relevant to other work studying graph games with modal logics, such as \cite{poisonlogic,argumentation,declan,graphgame}. Technically, the logic SLL has resemblances to several recent logics with model modifiers, such as \cite{movingarrow,changeoperator,satisfiabilityrelationchange}. Besides, instead of updating links, \cite{stepwiseremoval} considers a logic of stepwise point deletion, which sheds light on the long-standing open problem of how to axiomatize the sabotage-style modal logics. Moreover, \cite{stepwiseremoval} is also helpful to understand the complexity jumps between dynamic epistemic logics of model transformations and logics of freely chosen graph changes recorded in current memory. Another relevant line of research for this paper is epistemic logics. As mentioned already, one goal of our work is to avoid the Gettier problem. Similar to this, \cite{topology} uses the topological semantics to study the full belief.

Except what have been studied in this article, there are still various open problems deserving to be studied. From the logic point of view, Section \ref{subsec:application} shows that logic SLL is able to express the winning positions for players in finite games, but to capture those for infinite games, can SLL be expanded with some least-fixpoint operators? From the translation described in Definition \ref{def-translation} we know that SLL are effectively axiomatizable. However, is it possible to axiomatize the logic via a Hilbert-style calculus? In terms of games, we do not know the complexity of SLG, although Theorem \ref{theorem: model checking SLL} provides us with an upper bound. Besides, SLG includes exactly two players, and it is also meaningful to study the cases that are more general.

\vspace{.5cm}

\noindent{\bf{Acknowledgments.}} We thank Johan van Benthem, Fenrong Liu, Nina Gierasimczuk, Lena Kurzen, and Fernando R. Vel\'azquez-Quesada for their inspiring suggestions. We also wish to thank three anonymous LORI-VII referees for improvement
comments. Dazhu Li is supported by China Scholarship Council and the Major Program of the National Social Science Foundations of China [17ZDA026].
\bibliographystyle{plain}
\bibliography{LORI76}
\end{document}